\documentclass[12pt, draftclsnofoot, onecolumn]{IEEEtran}
\def\ol{\overline}
\usepackage{amsmath,amsfonts,amssymb}
\usepackage{amsthm}
\usepackage{algorithmic}
\usepackage{algorithm}
\usepackage{array}
\usepackage{enumitem}
\usepackage[caption=false,font=normalsize,labelfont=sf,textfont=sf]{subfig}
\usepackage{textcomp}
\usepackage{stfloats}
\usepackage{url}
\usepackage{bm}
\usepackage{verbatim}
\usepackage{graphicx}
\usepackage{cite}
\newtheorem{theorem}{Theorem}
\newtheorem{prop}{Proposition}
\newtheorem{lemma}{Lemma}
\newtheorem{corollary}{Corollary}

\begin{document}

\title{Signal Recovery From Product of Two Vandermonde Matrices}

\author{D\v{z}evdan Kapetanovi\'{c} \\ Huawei Technologies Sweden AB \\ Lund Research Center, Sweden}
        % <-this % stops a space



\maketitle

\begin{abstract}
In this work, we present some new results for compressed sensing and phase retrieval. For compressed sensing, it is shown that if the unknown $n$-dimensional vector can be expressed as a linear combination of $s$ unknown Vandermonde vectors (with Fourier vectors as a special case) and the measurement matrix is a Vandermonde matrix, exact recovery of the vector with $2s$ measurements and $O(\mathrm{poly}(s))$ complexity is possible when $n \geq 2s$. This result can be seen as a version of Prony's (or matrix pencil) method applied to the dual (time) domain. Based on this result, a new class of measurement matrices is presented from which it is possible to recover $s$-sparse $n$-dimensional vectors for $n \geq 2s$ with as few as $2s$ measurements and with a recovery algorithm of $O(\mathrm{poly}(s))$ complexity. In the second part of the work, these results are extended to the challenging problem of phase retrieval (compressed sensing with magnitude-only observations). The most significant discovery in this direction is that if the unknown $n$-dimensional vector is composed of $s$ frequencies with at least one being non-harmonic, $n \geq 4s - 1$ and we take at least $8s-3$ Fourier measurements, there are, remarkably, only \textit{two} possible vectors producing the observed measurement values and they are easily obtainable from each other. The two vectors can be found by an algorithm with only $O(\mathrm{poly}(s))$ complexity. This result is interesting when put in relation to recent work by Xu which shows that solving Fourier phase retrieval with $O(\mathrm{poly}(n))$ complexity proves that P = NP holds for the well known P vs NP problem. An immediate application of the new result is construction of a measurement matrix from which it is possible to recover almost all $s$-sparse $n$-dimensional signals (up to a global phase) from $O(s)$ magnitude-only measurements and $O(\mathrm{poly}(s))$ recovery complexity when $n \geq 4s - 1$. As a consequence, for almost all $s$-sparse $n$-dimensional signals, the proposed matrix construction significantly improves upon a result by Xu on finding a measurement matrix construction from which one can (provably) recover $s$-sparse signals with $O(s\log(n/s))$ phase-less measurements and 1-norm minimization. 
\end{abstract}


\section{Introduction}
\label{sec:intro}
We start by considering the conventional compressed sensing problem, which amounts to solving for $\bm x$ the following linear system 
\begin{align}
\label{eq:cs_main_eq}
\bm y = \bm A \bm x, \quad \bm x \in \mathbb{C}_{\leq s}^n,
\end{align} 
where $\mathbb{C}^n_{\leq s}$ is the space of $n$-dimensional complex vectors with at most $s$ elements being non-zero and $\bm A$ is a known $m\times n$ measurement matrix ($m$ is the number of \textit{measurements}). This problem has a rich history. Seminal works \cite{candes, donoho, candes2} and references therein show that, remarkably, there exist many matrices $\bm A$ such that $\bm x$ can be recovered with far less measurements than $n$. We will refer to \textit{exact recovery} of $\bm x$ as recovery of \textit{almost all} $\bm x$ (i.e. except for possibly some $\bm x$ that all belong to a set of measure zero in $\mathbb{C}_{\leq s}^n$); another common phrase is recovery \textit{with probability 1}. By \textit{probabilistic recovery} we refer to recovery of $\bm x$ at some probability strictly smaller than 1. The obvious 0-norm minimization that searches across all subsets of $s$ columns from $\bm A$ and inverts them to obtain a candidate solution is of exponential complexity. Fortunately, the work in \cite{candes_lp} shows that there exists a recovery algorithm for (\ref{eq:cs_main_eq}) with $O(\mathrm{poly}(n))$ complexity, where $O(.)$ stands for big O notation. More specifically, \cite{candes_lp} shows that for a given $\bm y = \bm A\bm x$, the unique solution to the convex 1-norm optimization problem
\begin{align}
\nonumber &\min_{\bm z} \|\bm z\|_1 \\
\nonumber &\textnormal{subject to} \\
&\bm y = \bm A \bm z \label{eq:cs_1_norm}
\end{align}
is $\bm z = \bm x$ if $\bm A$ satisfies the restricted isometry property (RIP). Some random constructions of $\bm A$ with $m = O(s\log(n/s))$ number of rows satisfy the RIP property with high probability; one example is a partial discrete Fourier transform (DFT) matrix with the $m$ rows chosen randomly from a $n\times n$ DFT matrix \cite{candes}. These existence and low recovery complexity results is what made compressed sensing a breakthrough technique with vast practical applications \cite{cs_practical}. Since random constructions of $\bm A$ only result in probabilistic recovery of $\bm x$, deterministic constructions that guarantee exact recovery are also of high interest. It is well known \cite{mpm, twc2023} that a $m \times n$ (partial) discrete Fourier transform (DFT) matrix, obtained from the first $m$ rows of a $n\times n$ DFT matrix, can be used as a measurement matrix with the matrix pencil method (MPM) as a recovery algorithm; in this case, it is enough with $m = 2s$ to exactly recover $\bm x$ with a recovery complexity of $O(\mathrm{poly}(s))$. However, the issue with this type of partial DFT matrix is that it is ill-conditioned, resulting in weak performance of MPM in presence of additive noise. Moreover, the low mutual coherence of the matrix \cite{partial_dft_coherence} also results in unstable performance in presence of additive noise in the constraint of (\ref{eq:cs_1_norm}). 
The work in \cite{bourgain} shows that for large $n$, there are deterministic constructions of $\bm A$ with $m = o(s^2)$ rows ($o(.)$ stands for Little o notation) and $n \geq m$ satisfying the RIP property which guarantees a low coherence and stability under noise.

The problem in (\ref{eq:cs_main_eq}) assumes full knowledge of the complex valued vector $\bm y$. A related problem is the recovery of $\bm x \in \mathbb{C}^n_s$ from \textit{phase-less} measurements
\begin{align}
\label{eq:phase_ret_main_eq}
\bm y = |\bm A\bm x|^2, \quad \bm x \in \mathbb{C}_s^n,
\end{align} 
where $|\bm x|^2$ is the vector of squared magnitudes of the elements in $\bm x$. The problem in (\ref{eq:phase_ret_main_eq}) is known as \textit{sparse phase retrieval} in the literature \cite{sparse_phase_ret1}-\cite{sparse_phase_ret3}. The recovery is allowed to be correct up to a common phase, i.e., obtaining $e^{i\alpha}\bm x$ for any $\alpha$ is considered as exact recovery. Clearly, due to the absence of phase information, this problem is significantly more challenging than the compressed sensing problem in (\ref{eq:cs_main_eq}) (which measurements we refer to as \textit{phase-aware} measurements); as such, more measurements are expected for exact recovery. Injectivity studies have shown that if $\bm A$ is a \textit{generic} matrix \cite{generic_phase_ret} then $m = 4s - 4$ measurements suffice for guaranteeing a unique solution. It is conjectured that the necessary lower bound is of the form $4s - O(s)$. The injectivity studies do not propose an efficient algorithm to check whether a certain $\bm A$ is generic nor do they provide an efficient recovery algorithm. A reformulation of (\ref{eq:cs_1_norm}) to (\ref{eq:phase_ret_main_eq}) results in the recovery algorithm
\begin{align}
\nonumber &\min_{\bm z} \|\bm z\|_1 \\
\nonumber&\textnormal{subject to} \\
\bm y &= |\bm A\bm x|^2 = |\bm A\bm z|^2. 
\label{eq:pr_1_norm}
\end{align}
Unfortunately, since (\ref{eq:pr_1_norm}) is a highly non-convex problem, it is no longer possible to prove that the unique solution to (\ref{eq:pr_1_norm}) is $\bm z = e^{i\alpha}\bm x$ for some $\alpha$. The work in \cite{xu_pr_1_norm} shows that if $\bm A$ is a circularly symmetric complex Gaussian matrix and $m = O(s\log(n/s))$, then solving (\ref{eq:pr_1_norm}) results in $\bm z = e^{i\alpha}\bm x$ with probability at least $1 - 2e^{-c_0 m}$ where $c_0$ is a constant; hence, only probabilistic recovery is guaranteed. When $\bm A$ is a discrete Fourier matrix, an iterative projection scheme is presented in \cite{moravec} that solves (\ref{eq:pr_1_norm}) with high probability (i.e., exact recovery cannot be guaranteed) when $m = O(s^2\log(n/s^2))$. However, the scheme is highly complex since it uses many iterations (200 000) for guaranteeing satisfactory performance and it assumes knowledge of $\|\bm x\|_1$ which is an impractical assumption. The complexity of solving (\ref{eq:pr_1_norm}) for a general $\bm A$ is at least $O(\mathrm{poly}(n))$ (probably significantly higher and possibly exponential), since the problem is at least as difficult as (\ref{eq:cs_1_norm}).

For general $\bm A$, PhaseLiftOff \cite{phase_lift_off} is a state-of-art algorithm that provides high probabilistic recovery with $O(s\log(n/s))$ measurements when elements in $\bm A$ are chosen from a complex Gaussian distribution. In order to avoid random constructions, PhaseCode \cite{phase_code} utilizes the theory of sparse-graph codes to provide a deterministic construction of $\bm A$ that is able to recover with high probability a large fraction of elements in $\bm x$ with $O(s)$ measurements. As a particular operating point, PhaseCode is able to recover a fraction of $1 - 10^{-7}$ non-zero components of $\bm x$ with $14s$ measurements and recovery probability $1 - O(1/m)$. Requiring a fraction of recovered elements close to 1 and a recovery probability close to 1 (i.e., almost exact recovery) drives the number of measurements to infinity; hence, as for PhaseLiftOff, exact recovery eludes PhaseCode as well.

In many practical instances of compressed sensing and phase retrieval, $\bm A$ is constrained to be a Fourier matrix \cite{candes, spm_2023, moravec} which in practice corresponds to an apparatus that takes frequency measurements of an underlying signal. The contributions in this work on compressed sensing show that if $\bm x$ is $s$-sparse in a Vandermonde basis (with Fourier basis as special case) and $n \geq 2s$, then it can be recovered exactly from $O(s)$ arbitrary Vandermonde measurements and $O(\mathrm{poly}(s))$ recovery complexity. In a way, this result is complementary to MPM; while MPM recovers signals that are sparse in the frequency domain with $O(\mathrm{poly}(s))$ recovery complexity, the method derived herein recovers signals that are sparse in the time domain with $O(\mathrm{poly}(s))$ recovery complexity. In case of Fourier phase retrieval, as shown in \cite{phase_ret_fourier}, the reason for why no exact $O(\mathrm{poly}(n))$ complexity algorithms have been proposed so far is because in doing so one solves the well known P vs NP problem, proving that P = NP holds. Interestingly, the contributions in this work on Fourier phase retrieval show that if the signal is $s$-sparse in a non-DFT basis with $n \geq 4s - 1$, and we take $O(s)$ measurements, then Fourier phase retrieval can be solved with $O(\mathrm{poly}(s))$ recovery complexity. 
\subsection{Contributions}
\label{subsec:contr}
Let $\bm V(\bm z)$ denote a Vandermonde matrix of $n$ rows parameterized by the elements from vector $\bm z = [z_0, \ldots, z_{m-1}]$ 
\begin{align}
\label{eq:vandermonde_def}
\bm V(\bm z) = \left[\begin{array}{ccc}1 & \ldots & 1 \\ z_0 & \ldots & z_{m-1} \\ \vdots & \ddots & \vdots \\ z_0^{n-1} & \ldots & z_{m-1}^{n-1} \end{array} \right].
\end{align}
Let $\mathbb{T} = \{z ; z \in \mathbb{C}, |z| = 1\}$ be the set of all points on the complex unit circle and $(.)^T$ the matrix transpose operator. The term \textit{almost all} $\bm x \in \mathbb{A}$, for some set $\mathbb{A}$, is referring to a set that equals $\mathbb{A}$ up to a null set (set of measure zero in $\mathbb{A}$). When $z_k \in \mathbb{T}$, $0 \leq k \leq m-1$, the Vandermonde matrix becomes a Fourier matrix, which we for clarity henceforth denote as $\bm F(\bm z)$ (with DFT matrix as a special case when $z_k$ are $n$th roots of unity and $m = n$).  For Vandermonde and Fourier measurement matrices, we have the following main results
\begin{enumerate}[label=\arabic*.]
\item {\label{it:phase_ret} When $\bm A = \bm V(\bm z)^T$, $\bm x$ is of the form $\bm x = \bm V(\bm \theta)\bm g$ for unknown $\bm \theta \in \mathbb{C}^s, \bm g \in \mathbb{C}^s$, and $n \geq 2s$, we have exact recovery of $\bm \theta$ and $\bm g$ from $\bm y = \bm A \bm x$ with $O(\mathrm{poly}(s))$ recovery complexity and
\begin{enumerate}[label=\alph*.]
\item When the elements in $\bm z$ are $n$th roots of unity, $m = 2s$ measurements suffice.
\item For almost all $\bm z \in \mathbb{C}^m$, $m \geq 3s$ suffices.
\end{enumerate}
}
\item {\label{it:phaseless_ret} Assume $\bm A = \bm F(\bm z)^T$, $\bm x$ is of the form $\bm x = \bm F(\bm \theta)\bm g$, with $\bm \theta \in \mathbb{T}^s$ and $\bm g \in \mathbb{C}^s$ being unknown, and $n \geq 4s - 1$. Let $\bm y = |\bm A \bm x|^2$. 
\begin{enumerate}[label=\alph*.]
\item{\label{it:dft} If the elements in $\bm \theta$ are $n$th roots of unity ($\bm F(\bm \theta)$ equals a subset of $s$ columns from a $n \times n$ DFT matrix), we have exact recovery of $\bm \theta$ and $|\bm g|$ (up to a global real positive scalar) from $\bm y$ and $\bm A$ with $O(\mathrm{poly}(s))$ recovery complexity. There are $2^{s-1}$ different $\bm g$ (up to a global phase), all having the same value of $|\bm g|$, producing the same $\bm y$.}

\item{\label{it:non-dft} If not all elements in $\bm \theta$ are $n$th roots of unity and for almost all $\bm g$
\begin{enumerate}[label=\roman*.]
\item When the elements in $\bm z$ are $n$th roots of unity and $m = 4s-1$, we have exact recovery of $\bm \theta$ and $|\bm g|$ (up to a global real positive scalar) from $\bm y$ and $\bm A$ with $O(\mathrm{poly}(s))$ recovery complexity. There are $2^{s-1}$ different $\bm g$ (up to a global phase), all having the same value of $|\bm g|$, producing the same $\bm y$. 
\item For almost all $\bm z \in \mathbb{T}^m$ with $m \geq 8s-3$, we have exact recovery of $\bm \theta$ and $|\bm g|$ (up to a global real positive scalar) from $\bm y$ and $\bm A$ with $O(\mathrm{poly}(s))$ recovery complexity. There are only two different $\bm g$ (up to a global phase), both having the same value of $|\bm g|$, that produce the same $\bm y$; moreover, the two solutions are easily obtainable from each other. 
\end{enumerate}
}
\end{enumerate}
}
\end{enumerate}
The result in \ref{it:phase_ret} is a complementary result to MPM since it recovers signals that are sparse in the (time) domain that is dual to the (frequency) domain of MPM. Moreover, it shows that recovery is guaranteed even with arbitrary samples, in contrast to MPM that only allows consecutive integer samples. 

For phase-less measurements, the result in \ref{it:phaseless_ret} states that $\bm \theta$ and $c|\bm g|$, for unknown real positive scalar $c$, can be found with $O(s)$ measurements and $O(\mathrm{poly}(s))$ recovery complexity. Thus, in scenarios where only the frequency components $\bm \theta$ and their \textit{relative} magnitudes are sought for, the needed number of measurements and the recovery complexity are essentially lowest possible. Moreover, the different vectors $\bm g$ that produce the same measurements $\bm y$ all have same magnitudes of elements on same positions, meaning that only the phases of the elements differ.

On the other hand, when it comes to recovering $\bm g$ from phase-less measurements, \ref{it:phaseless_ret}\ref{it:dft} is in line with the work in \cite{phase_ret_fourier}, which showed that solving Fourier phase retrieval with $O(\mathrm{poly}(n))$ recovery complexity for general $\bm x$ proves that P = NP holds for the well known P vs NP problem. As seen, \ref{it:phaseless_ret}\ref{it:dft} states that recovering an oversampled signal composed of $s$ harmonic frequencies with an oversampling factor of at least 4 ($n \geq 4s - 1$) still results in an exponential recovery complexity in $s$ since there are $2^{s-1}$ different $\bm g$ giving rise to the measurements in $\bm y$. Therefore, this work does \textit{not} prove that P = NP, since it doesn't provide polynomial recovery complexity for \textit{all} $\bm x$. Interestingly, as noted above, the $2^{s-1}$ different vectors $\bm g$ only differ in the phases of their elements, as if there are only two choices for the phase of each element (once the phase of one element is fixed).

Conversely, \ref{it:phaseless_ret}\ref{it:non-dft} states that when $\bm x$ is composed of $s$ frequencies with at least one being non-harmonic, where $n \geq 4s - 1$ and the number of measurements $m \geq 8s - 3$, there are only \textit{two} possible solutions (one of them is of course the true vector $\bm x$) and they can be found with $O(\mathrm{poly}(s))$ complexity. Obviously, to find which of the two possible solutions is $\bm x$, one can use additional side information as in \cite{phase_ret_fourier} by assuming knowledge of one element from $\bm x$. Thus, Fourier phase retrieval \textit{is} solvable with $O(s)$ measurements and $O(\mathrm{poly}(s))$ recovery complexity for $n$-dimensional signals containing (at most) $\lfloor n/4 \rfloor$ frequencies with at least one frequency being non-harmonic.

Removing the assumption of $\bm x$ being spanned by a Vandermonde or Fourier basis, we look at a measurement matrix construction $\bm A$ for recovering $\bm x \in \mathbb{C}_{\leq s}^n$. Let $;$ denote the vertical concatenation of matrices. In this regard, the main contributions of this work can be summarized as follows
\begin{enumerate}[label=\arabic*.]
\setcounter{enumi}{2}
\item {\label{it:vander_product} When $\bm A = \bm V(\bm z)^T \bm V(\bm \theta)$, $\bm z \in \mathbb{C}^m$, $\bm \theta \in \mathbb{C}^n$, and $n \geq 2s$, we get exact recovery of $\bm x \in \mathbb{C}_{\leq s}^n$ from $\bm y = \bm A \bm x$ with $O(\mathrm{poly}(s))$ recovery complexity and
\begin{enumerate}[label=\alph*.]
\item When elements in $\bm z$ are chosen as $n$th roots of unity and for almost all $\bm \theta$, $m = 2s$ measurements suffice.
\item For almost all $\bm z$ and $\bm \theta$, $m \geq 3s$ suffices.
\end{enumerate}
}
\end{enumerate}
\begin{enumerate}[label=\arabic*.]
\setcounter{enumi}{3}
\item {Let $\bm y = |\bm A\bm x|^2$, where $\bm x \in \mathbb{C}_{\leq s}^n$, $\bm z \in \mathbb{T}^m$, $\bm \theta \in \mathbb{T}^n$ and $n \geq 4s-1$. \label{it:fourier_prod}
\begin{enumerate}[label=\alph*.]
\item {\label{it:no_rand_measure} Assume $\bm A = \bm F(\bm z)^T \bm F(\bm \theta)$. 
\begin{enumerate}[label=\roman*.]
\item When $m = 4s - 1$, the elements in $\bm z$ are $n$th roots of unity and for almost all $\bm \theta$, we have exact recovery of the support of $\bm x \in \mathbb{C}_{\leq s}^n$ and $|\bm x|$ (up to a global real positive scalar) from $\bm y$ and $\bm A$. There are $2^{s-1}$ different $\bm x$ (all with the same value of $|\bm x|$) producing the same $\bm y$.
\item For almost all $\bm z$ and $\bm \theta$ with $m \geq 8s - 3$, we have exact recovery of the support of $\bm x \in \mathbb{C}_{\leq s}^n$ and $|\bm x|$ (up to a global real positive scalar) from $\bm y$ and $\bm A$. There are only two different $\bm x$ (with the same value of $|\bm x|$) producing the same $\bm y$ and they are easily obtainable from each other.
\end{enumerate}
}

\item {\label{it:rand_measure} Let $\bm A = [\bm F(\bm z)^T \bm F(\bm \theta); \bm a^T]$ where $\bm a \in \mathbb{C}^n$ is some random vector. Then, 
\begin{enumerate}[label=\roman*.]
\item When the elements in $\bm z$ are $n$th roots of unity and for almost all $\bm \theta$, $m = 4s-1$ measurements suffice to exactly recover $\bm x \in \mathbb{C}_{\leq s}^n$ from $\bm y$ and $\bm A$ with a recovery complexity of $O(\mathrm{poly}(s)2^s)$.
\item For almost all $\bm z$ and $\bm \theta$, $m \geq 8s - 3$ suffices to exactly recover $\bm x \in \mathbb{C}_{\leq s}^n$ from $\bm y$ and $\bm A$ with $O(\mathrm{poly}(s))$ recovery complexity.
\end{enumerate}}
\end{enumerate}
}
\end{enumerate}
The result in \ref{it:vander_product} provides a new explicit construction of $\bm A$ with $m = O(s)$ rows (even achieving the lower bound $2s$) which produces a unique $\bm x$ in (\ref{eq:cs_main_eq}) that can be recovered with $O(\mathrm{poly}(s))$ complexity (or $O(\mathrm{poly}(n))$ complexity by solving (\ref{eq:cs_1_norm})); even more, a whole new \textit{class} of matrices $\bm A$ is presented that achieves this.

The results in \ref{it:fourier_prod} are obtained rather easily from \ref{it:phaseless_ret} which can be seen from the similarity of the statements.
The result in \ref{it:fourier_prod}\ref{it:rand_measure} solves problem 3.5 in \cite{xu} for \textit{almost all} $\bm x \in \mathbb{C}_{\leq s}^n$ and significantly improves upon it. Namely, that problem asks for an explicit construction of $\bm A$ with $m = O(s\log(n/s))$ such that solving (\ref{eq:pr_1_norm}) results in recovery of \textit{all} $\bm x \in \mathbb{C}_{\leq s}^n$. The combination of the deterministic construction $\bm F(\bm z)^T \bm F(\bm \theta)$ and a single random measurement as described in \ref{it:fourier_prod}\ref{it:rand_measure} shows that after only $4s$ measurements there is a unique solution to $\bm x$ (up to a global phase) in (\ref{eq:phase_ret_main_eq}) for almost all $\bm x \in \mathbb{C}_{\leq s}^n$; thus even (\ref{eq:pr_1_norm}) will find it. Noteworthy is that the number of measurements $4s$ in the proposed construction almost achieves the injectivity lower bound of $4s - 4$.
Furthermore, taking $8s-2$ measurements instead of $4s$, we are able to reduce the recovery complexity to $O(\mathrm{poly}(s))$ with a novel recovery algorithm, drastically reducing the recovery complexity for large scale problems (large $n$) compared to (\ref{eq:pr_1_norm}).
\subsection{Organization}
The rest of the article is structured as follows. In Section \ref{sec:phase_measurements} we present two results described in Section \ref{sec:intro} for phase-aware measurements and their derivations. Section \ref{sec:phaseless_msrmt} presents three results described in Section \ref{sec:intro} for phase-less measurements and their derivations. 
\section{Phase-aware Measurements}
\label{sec:phase_measurements}
We start by formulating the Vandermonde measurement result outlined in Section \ref{sec:intro}.
\par\noindent\rule{\textwidth}{0.4pt}
\textbf{R1}: Assume an $n$-dimensional vector $\bm x = \bm V(\bm \theta)\bm g$, where $\bm \theta \in \mathbb{C}^s$, $\bm g \in \mathbb{C}^s$ and $n \geq 2s$. There are samples $\bm z \in \mathbb{C}^m$, with $m \geq 2s$, such that almost all $\bm \theta$ and $\bm g$ (and thus almost all $\bm x$) can be exactly recovered from
\begin{align}
\label{eq:phase_msrmt_vandermonde}
\bm y = \bm V(\bm z)^T \bm x = \bm V(\bm z)^T \bm V(\bm \theta)\bm g
\end{align}
with $O(\mathrm{poly}(s))$ recovery complexity.
\par\noindent\rule{\textwidth}{0.4pt}
Since $\theta_j$, $0 \leq j \leq s-1$, are assumed to be continuous, \textbf{R1} achieves grid-less compressed sensing \cite{gridless_comp_sens} with the assumptions that the measurement matrix is $\bm V(\bm z)^T$ and that the $n$-dimensional complex vector $\bm x$ is inside an $s$-dimensional complex subspace (with $n \geq 2s$) spanned by a $n\times s$ Vandermonde basis matrix. A version of \textbf{R1} in the specific case of Fourier matrices is presented in \cite{twc2023}. Hence, \textbf{R1} extends the results in \cite{twc2023} to general Vandermonde matrices. 

\textbf{R1} is derived by first solving for the unknown vector $\bm\theta$. Once $\bm \theta$ is known, it is proved that there is a unique $\bm g$ satisfying (\ref{eq:phase_msrmt_vandermonde}) and it can, e.g., be obtained as $\bm g = \left(\bm V(\bm z)^T \bm V(\bm \theta)\right)^{\dagger}\bm y$ where $(.)^{\dagger}$ denotes the matrix pseudo-inverse. Deriving \textbf{R1} in this particular order gives us  
\par\noindent\rule{\textwidth}{0.4pt}
\textbf{R2}: Assume $n \geq 2s$ and $m \geq 2s$. For infinitely many choices of $\bm z \in \mathbb{C}^m, \bm \theta \in \mathbb{C}^n$, one can recover all $\bm x \in \mathbb{C}_{\leq s}^n$ from $\bm y = \bm V(\bm z)^T \bm V(\bm \theta) \bm x$ with $O(\mathrm{poly}(s))$ recovery complexity.
\par\noindent\rule{\textwidth}{0.4pt}
\subsection{Derivation of \textbf{R1}}
As observed in \cite{twc2023}, element $y_j$ in \textbf{R1} can be seen as a sample $y_j = f(z_j)$ of the rational function 
\begin{align}
f(z) &= \sum_{l=0}^{s-1}g_l\sum_{j=0}^{n-1}(z\theta_l)^j = \sum_{l=0}^{s-1}g_l\frac{(z\theta_l)^n - 1}{z\theta_l - 1} = \frac{u(z)}{v(z)} \label{eq:f} \\
u(z) &\stackrel{\triangle}{=} \sum_{l=0}^{s-1}g_l((z\theta_l)^n-1)\prod_{\substack{i=0 \\ i\neq l \\g_i \neq 0}}^{s-1}(z\theta_i-1) \label{eq:u} \\
v(z) &\stackrel{\triangle}{=} \prod_{\substack{l=0 \\ g_l\neq 0}}^{s-1}(z\theta_l-1) \label{eq:v}.
\end{align}
The polynomial $u(z)$ can be further expressed as
\begin{align}
\label{eq:u_exp}
\nonumber u(z) &= z^n\hat{u}(z) + \tilde{u}(z) \\
\nonumber\hat{u}(z) &\stackrel{\triangle}{=} \sum_{l=0}^{s-1}g_l\theta_l^nt_l(z) \\
\nonumber\tilde{u}(z) &\stackrel{\triangle}{=} -\sum_{l=0}^{s-1}g_lt_{l}(z) \\
t_{l}(z) &\stackrel{\triangle}{=} \prod_{\substack{i=0 \\ i\neq l \\g_i \neq 0}}^{s-1}(z\theta_i-1).
\end{align}
Obviously, if some elements in $\bm g$ are zero then $s$ is just reduced to a smaller value. Thus, without loss of generality (WLOG), we will at first assume that $\bm g$ contains no zeros and later, when necessary, comment on the impact of an unknown reduced $s$ on some developments in this work. Furthermore, throughout this work, we also assume that the elements in $\bm \theta$ are distinct (otherwise at least two columns in $\bm V(\bm \theta)$ are identical) and that none of them is zero. 

Next, some results on the introduced polynomials are given that will be important for this work. 
%Clearly, $t_j(z), 0 \leq j \leq s-1$, are (scaled) Lagrange polynomials which are well known to be linearly independent. For completeness, we present a proof of this fact.
\begin{lemma}
\label{lemma:ind_t_pol}
Assume that $\bm g \in \mathbb{C}^s$ has no zeros. The polynomials $t_{0}(z),\ldots,t_{s-1}(z)$ are linearly independent.
\end{lemma}
\begin{proof}
Since no element in $\bm g$ is zero, $t_j(z), 0 \leq j \leq s-1$, are all distinct and of degree $s-1$. Assume that these polynomials are linearly dependent, so that we can write $ \sum_{j=0}^{s-1}c_jt_{j}(z) = 0$ for some coefficients $c_0,\ldots,c_{s-1}$ not all being 0. Since none of $t_j(z), 0 \leq j \leq s-1$, are zero, at least two of $c_j, 0 \leq j \leq s-1$, are non-zero. Assume then that $c_j \neq 0$ and $c_k \neq 0$ for some $j \neq k$. Hence, we can write $t_k(z) = -\frac{1}{c_k}\sum_{\substack{j = 0 \\ j \neq k}}^{s-1}c_jt_j(z)$. Since $t_{k}(\theta_k^{-1}) \neq 0$ and $t_{j}(\theta_k^{-1}) = 0$ for $j \neq k$, it follows that $0 \neq t_k(\theta_k^{-1}) = -\frac{1}{c_k}\sum_{\substack{j = 0 \\ j \neq k}}^{s-1}c_jt_j(\theta_k^{-1}) = 0$, a contradiction, implying that $t_{0}(z),\ldots,t_{s-1}(z)$ are linearly independent. 
\end{proof}
\begin{lemma}
\label{lemma:u_tilde_distinct_roots}
For almost all $\bm g \in \mathbb{C}^s$ with $s$ non-zero elements, $\tilde{u}(z)$ has $s-1$ roots with multiplicity one.
\end{lemma}
\begin{proof}
For a $\bm g \in \mathbb{C}^s$ that contains no zeros, the polynomials $t_{0}(z),\ldots,t_{s-1}(z)$ are linearly independent from Lemma \ref{lemma:ind_t_pol}. Varying $\bm g$ across $\mathbb{C}^s$ makes $\tilde{u}(z)$ span all polynomials of degree at most $s-1$. Polynomials of degree smaller than $s-1$ correspond to a complex subspace of dimension lower than $s$. Thus, for almost all $\bm g$, $\tilde{u}(z)$ has degree $s-1$. Each polynomial of degree $s-1$ can be identified (up to a complex scalar) with the vector of its roots $\bm r = [r_0,\ldots,r_{s-2}] \in \mathbb{C}^{s-1}$; in case of multiple roots, corresponding elements in $\bm r$ are the same. Polynomials with roots that have multiplicity larger than one can therefore be represented as a union of subspaces in $\mathbb{C}^{s-1}$ of dimensions lower than $s-1$. Hence, for almost all $\bm r \in \mathbb{C}^{s-1}$ (and thus almost all $\bm g \in \mathbb{C}^s$), $\tilde{u}(z)$ has $s-1$ roots with multiplicity one.
\end{proof}
\begin{lemma}
\label{lemma:u_hat_tilde_roots}
Given a $\bm \theta$ such that $\theta_j^n, 0 \leq j \leq s-1$, are not all equal. Then there exists polynomials $\tilde{u}(z)$, $\hat{u}(z)$ that do not have any roots in common.
\end{lemma}
\begin{proof}
From (\ref{eq:u_exp}) it follows that 
\begin{align}
\label{lemma_eq:g_exp}
g_j = -\frac{\tilde{u}(\theta_j^{-1})}{t_{j}(\theta_j^{-1})}, \quad 0 \leq j \leq s-1.
\end{align}
Hence, the $g_j, 0 \leq j \leq s-1$, (and thus $\tilde{u}(z)$, $\hat{u}(z)$ as well) are uniquely determined from the $s$ samples $\tilde{u}(\theta_j^{-1}), 0 \leq j \leq s-1$. The identity in (\ref{lemma_eq:g_exp}) gives
\begin{align}
\label{lemma_eq:g_exp_u_exp}
\nonumber \tilde{u}(z) = -\sum_{j=0}^{s-1}\frac{\tilde{u}(\theta_j^{-1})}{t_{j}(\theta_j^{-1})}t_{j}(z) \\
\hat{u}(z) = -\sum_{j=0}^{s-1}\frac{\tilde{u}(\theta_j^{-1})}{t_{j}(\theta_j^{-1})}\theta_j^nt_{j}(z).
\end{align}
 Choose $\tilde{u}(\theta_j^{-1}) = 1, 0 \leq j \leq s-1$, which gives $\tilde{u}(z) = 1$. Clearly, $\hat{u}(z) \neq 0$ for this choice since the assumption on $\bm \theta$ is that it has no zeros. 
Assume then instead that $\hat{u}(z)$ is a constant polynomial, $\hat{u}(z) = c$, for some complex scalar $c$. Then, from (\ref{lemma_eq:g_exp_u_exp})
\begin{align}
\label{lemma_eq:u_tilde_hat_diff}
\hat{u}(z) - c\tilde{u}(z) = \sum_{j=0}^{s-1}\frac{c - \theta_j^n}{t_{j}(\theta_j^{-1})}t_{j}(z) = 0
\end{align}
which holds if and only if $\theta_j^n = c, 0 \leq j \leq s-1$, contradicting the assumption on $\bm \theta$. Hence, $\hat{u}(z)$ has at least degree one.
\end{proof}
Next, we show that for almost all $\bm g$, $\tilde{u}(z)$ and $\hat{u}(z)$ share no common roots if and only if $\theta_j^n, 0 \leq j \leq s-1$, are not all equal. To do this, we utilize Lemma \ref{lemma:u_hat_tilde_roots} and the well known result that two polynomials $p(x), q(x)$ of degrees $m, n$, respectively, share a common root if and only if their \textit{resultant}  $\mathrm{res}(p,q)$ is zero, where
\begin{align}
\label{eq:resultant}
\mathrm{res}(p,q) = \left|\begin{array}{cccccccc}p_m & p_{m-1} & \ldots & p_0 & 0 & \ldots & \ldots & 0 \\
0 & p_m & p_{m-1} & \ldots & p_0 & 0 & \ldots & 0 \\
\vdots & & & \ddots & & \ddots & & \vdots \\
0 & \ldots & \ldots & 0 & p_m & p_{m-1} & \ldots & p_0 \\
q_n & q_{n-1} & \ldots & q_0 & 0 & \ldots & \ldots & 0 \\
0 & q_n & q_{n-1} & \ldots & q_0 & 0 & \ldots & 0 \\
\vdots & & & \ddots & & \ddots & & \vdots \\
0 & \ldots & \ldots & 0 & q_n & q_{n-1} & \ldots & q_0
\end{array}\right|
\end{align}
and $|.|$ denotes the determinant. The matrix in (\ref{eq:resultant}) has dimension $m+n$, where coefficients of $p(z)$ appear in $n$ rows and coefficients of $q(z)$ appear in $m$ rows. 
\begin{theorem}
\label{thm:u_hat_tilde_roots_general}
$\tilde{u}(z)$ and $\hat{u}(z)$ share no common roots for almost all $\bm g \in \mathbb{C}^s$ if and only if $\bm \theta$ is such that $\theta_j^n, 0 \leq j \leq s-1$, are not all equal.
\end{theorem}
\begin{proof}
Note from (\ref{eq:u_exp}) that if $\theta_j^n, 0 \leq j \leq s-1$, are all equal then $\tilde{u}(z) = \hat{u}(z)$ and thus they have the same roots. To finish the proof, we now show that when not all $\theta_j^n, 0 \leq j \leq s-1$, are equal, then $\tilde{u}(z)$ and $\hat{u}(z)$ share no common roots for almost all $\bm g$.

From (\ref{eq:resultant}), it follows that the elements of the matrix in $\mathrm{res}(\tilde{u},\hat{u})$ are a linear combination of the elements in $\bm g$. Thus, $r(\bm g) = \mathrm{res}(\tilde{u},\hat{u})$ is a \textit{multivariate polynomial} in the elements of $\bm g$. 
Denote by $\bm g^1$ the $\bm g$ that produces the polynomials of Lemma \ref{lemma:u_hat_tilde_roots}, i.e., $\tilde{u}(z) = 1$ and $\hat{u}(z)$ is of degree at least one. For these polynomials, the matrix in (\ref{eq:resultant}) is an identity matrix and thus $r(\bm g^1) = 1$, showing that $r(\bm g)$ is not the zero polynomial. The set $\mathcal{S} = \{\bm g | \bm g \in \mathbb{C}^s, r(\bm g) = 0\}$ is either empty or is a hypersurface of dimension $s-1$ in the ambient space $\mathbb{C}^s$. In both cases, it has measure zero in $\mathbb{C}^s$, implying that $r(\bm g) \neq 0$ for almost all $\bm g$, which proves the theorem.
\end{proof}

The proof for the bounds on $n$ and $m$ in \cite{twc2023} holds here as well, and thus we have the necessary conditions
\begin{theorem}
\label{eq:lbs}
\cite{twc2023} If \textbf{R1} has a (unique) solution then $n \geq 2s$ and $m \geq 2s$. 
\end{theorem} 

Denote by $y_j^+ = f^+(z_j)$, $0 \leq j \leq m-1$, the observed measurements which stem from polynomials $v^+(z), \hat{u}^+(z),  \tilde{u}^+(z)$ that correspond to realizations $\bm \theta = \bm \theta^+, \bm g = \bm g^+$.
Furthermore, denote by $\bm v, \hat{\bm u}, \tilde{\bm u}$ the coefficient vectors of polynomials $v(z), \hat{u}(z), \tilde{u}(z)$ with degrees $s, s-1, s-1$, respectively. Consider the set of linear equations
\begin{align}
\label{eq:phase_poly_samples_eqs}
y_j^+v(z_j) - z_j^n\hat{u}(z_j) - \tilde{u}(z_j) &= 0, \quad 0 \leq j \leq m-1 
\end{align}
in $\bm v, \hat{\bm u}, \tilde{\bm u}$. Let $\bm 0_n$ denote the vector of $n$ zeros (in case its dimension is implicit from the context, we will just denote it as $\bm 0$).
The equation system in (\ref{eq:phase_poly_samples_eqs}) corresponds to solving the linear equation system 
\begin{align}
\label{eq:matrix_eqs}
\nonumber \bm A(\bm z)\bm w &= \bm 0 \\
\nonumber &\textnormal{where} \\
\bm A(\bm z) &= \left[\begin{array}{ccccccccc} y_0^+z_0^s  & \ldots & y_0^+ & -z_0^{n+s-1} & \ldots & -z_0^n & -z_0^{s-1} & \ldots & -1 \\ \vdots & & \vdots & \vdots & & \vdots & \vdots & & \vdots\\
y_{m-1}^+z_{m-1}^s & \ldots & y_{m-1}^+ & -z_{m-1}^{n+s-1} & \ldots & -z_{m-1}^n & -z_{m-1}^{s-1} & \ldots & -1 \end{array}\right]
\end{align}
with $\bm A(\bm z)$ being of dimension $m \times (3s+1)$ and $\bm w = [\bm v; \hat{\bm u}; \tilde{\bm u}]$. Assume that $S$ values from $g_0^+,\ldots,g_{s-1}^+$ are non-zero; hence, $v^+(z), \hat{u}^+(z), \tilde{u}^+(z)$ are of degrees $S, S-1, S-1$, respectively. Since $v(z) = v^+(z), \hat{u}(z) = \hat{u}^+(z), \tilde{u}(z) = \tilde{u}^+(z)$ solve (\ref{eq:phase_poly_samples_eqs}), $\bm w = \bm w^+$ with 
\begin{align}
\label{eq:wp_phase}
\bm w^+ = [\bm 0_{s-S}; \bm v^+; \bm 0_{s-S}; \hat{\bm u}^+; \bm 0_{s-S}; \tilde{\bm u}^+]
\end{align}
solves (\ref{eq:matrix_eqs}) and vice versa. This implies that $\bm A(\bm z)$ cannot be full rank and thus it has a null space of at least one dimension. If $\bm A(\bm z)$ has a one-dimensional null space, meaning that all solutions to (\ref{eq:matrix_eqs}) are of the form $\bm w = c\bm w^+$ for any complex scalar $c$, then (\ref{eq:phase_poly_samples_eqs}) has the polynomial solution $v(z) = cv^+(z)$, $\hat{u}(z) = c\hat{u}^+(z)$, $\tilde{u}(z) = c\tilde{u}^+(z)$ and vice versa.

In the special case of $z_j = e^{i2\pi j/n}e^{i\gamma/n}$ and $m \leq n$, i.e., rotated $n$th roots of unity which we for brevity refer to as \textit{shifted harmonics} (or \textit{$\gamma$-shifted harmonics} if needing to explicitly state $\gamma$) since the same $n$th root is used throughout this work, we have $z_j^n = e^{i\gamma}$ so that (\ref{eq:phase_poly_samples_eqs}) becomes
\begin{align}
\label{eq:phase_poly_dft_samples_eqs}
y^+(z_j)v(z_j) - e^{i\gamma}\hat{u}(z_j) - \tilde{u}(z_j) &= 0, \quad 0 \leq j \leq m-1
\end{align}
which corresponds to solving the linear system
\begin{align}
\label{eq:matrix_dft_eqs}
\nonumber \bm B(\bm z)\bm w &= \bm 0 \\
\nonumber &\textnormal{where} \\
\bm B(\bm z) &= \left[\begin{array}{cccccc} y_0^+z_0^s  & \ldots & y_0^+ & -z_0^{s-1} & \ldots & -1 \\ & \vdots & & & \vdots &\\
y_{m-1}^+z_{m-1}^s & \ldots & y_{m-1}^+ & -z_{m-1}^{s-1} & \ldots & -1 \end{array}\right]
\end{align}
with $\bm B(\bm z)$ being of dimension $m \times (2s+1)$ and $\bm w = [\bm v; e^{i\gamma}\hat{\bm u} + \tilde{\bm u}]$. As with (\ref{eq:phase_poly_samples_eqs}) and (\ref{eq:matrix_eqs}), a one-dimensional null space of $\bm B(\bm z)$ gives all solutions to (\ref{eq:matrix_dft_eqs}) as $\bm w = c\bm w^+$ for $c \in \mathbb{C}$ with 
\begin{align}
\label{eq:wp_dft_phase}
\bm w^+ = [\bm 0_{s-S}; \bm v^+; \bm 0_{s-S}; e^{i\gamma}\hat{\bm u}^+ + \tilde{\bm u}^+],
\end{align} 
which is equivalent to the polynomial solution $v(z) = cv^+(z)$, $e^{i\gamma}\hat{u}(z) + \tilde{u}(z) = c(e^{i\gamma}\hat{u}^+(z) + \tilde{u}^+(z))$ of (\ref{eq:phase_poly_dft_samples_eqs}) and vice versa.

Next, we generalize Corollary 2 in \cite{twc2023}.
\begin{theorem}
\label{thm:2k_dft}
Assume that $\bm g^+$ has no zeros. Furthermore, assume that $z_j^n = e^{i\gamma}$ for $0 \leq j \leq 2s-1$ and $(\theta_k^+)^n \neq e^{-i\gamma}$ for $0 \leq k \leq s-1$. Solving (\ref{eq:phase_poly_dft_samples_eqs}) (or equivalently (\ref{eq:matrix_dft_eqs})) results in $v(z) = cv^+(z)$ and $e^{i\gamma}\hat{u}(z) + \tilde{u}(z) = c(e^{i\gamma}\hat{u}^+(z) + \tilde{u}^+(z))$ for some complex scalar $c$.
\end{theorem}
\begin{proof}
Since $z_j$, $0 \leq j \leq 2s - 1$, are shifted harmonics and $(\theta_k^+)^n \neq e^{-i\gamma}$, it follows that $v^+(z_j) \neq 0$ and $t_{k}((\theta_k^+)^{-1}) \neq 0$ for $0 \leq j \leq 2s - 1$, $0 \leq k \leq s-1$. Inserting $y^+_j = \frac{u^+(z_j)}{v^+(z_j)}$, $0 \leq j \leq 2s-1$, into (\ref{eq:phase_poly_dft_samples_eqs}) and using $v^+(z_j) \neq 0$ we get
\begin{align}
\label{eq_proof:poly_dft_eqs}
\nonumber \frac{u^+(z_j)}{v^+(z_j)}v(z_j) - e^{i\gamma}\hat{u}(z_j) - \tilde{u}(z_j) &= 0, \quad 0 \leq j \leq 2s-1 \\
(e^{i\gamma}\hat{u}(z_j) + \tilde{u}(z_j))v^+(z_j) &= u^+(z_j)v(z_j), \quad 0 \leq j \leq 2s-1.
\end{align}
Note from (\ref{eq:u_exp}) that $u^+(z_j) = e^{i\gamma}\hat{u}^+(z_j) + \tilde{u}^+(z_j)$. Since the polynomials on both sides of the identity in (\ref{eq_proof:poly_dft_eqs}) are of degree $2s-1$, the identity theorem for polynomials implies that
\begin{align}
\label{eq_proof:poly_all_z}
(e^{i\gamma}\hat{u}(z) + \tilde{u}(z))v^+(z) = (e^{i\gamma}\hat{u}^+(z) + \tilde{u}^+(z))v(z), \forall z
\end{align}
The roots of $v(z)$ are $(\theta_0^+)^{-1}, \ldots, (\theta_{s-1}^+)^{-1}$. Inserting $z = (\theta_k^+)^{-1}$, $0 \leq k \leq s-1$, into $e^{i\gamma}\hat{u}(z) + \tilde{u}(z)$ gives $e^{i\gamma}\hat{u}((\theta_k^+)^{-1}) + \tilde{u}((\theta_k^+)^{-1}) = g_k^+t_{k}^+((\theta_k^+)^{-1})(e^{i\gamma}(\theta_k^+)^n - 1)$ which is not 0 since $g_k^+ \neq 0$,  $(\theta_k^+)^n \neq e^{-i\gamma}$ and $t_{k}^+((\theta_k^+)^{-1}) \neq 0$. Hence, from (\ref{eq_proof:poly_all_z}) it follows that the roots of $v(z)$ are roots of $v^+(z)$ and thus $v(z) = cv^+(z)$ for some complex scalar $c$. Inserting this relation back into (\ref{eq_proof:poly_all_z}) and dividing away $v^+(z)$ we obtain $e^{i\gamma}\hat{u}(z) + \tilde{u}(z) = c(e^{i\gamma}\hat{u}^+(z) + \tilde{u}^+(z))$.
\end{proof}
Hence, Theorem \ref{thm:2k_dft} shows that we can recover $\bm \theta^+$ after $2s$ measurements with shifted harmonics; obviously, the theorem also holds if there are more than $2s$ shifted harmonics. Note that there is no constraint on $\theta_k^+, 0 \leq k \leq s-1,$ to lie inside $\mathbb{T}$ as in \cite{twc2023}. For a $\bm z$ satisfying the assumptions of Theorem \ref{thm:2k_dft}, almost all  $\bm \theta^+$ and $\bm g^+$ satisfy the assumptions of Theorem \ref{thm:2k_dft}. Similarly, for any given $\bm \theta^+$ and $\bm g^+$ (with $s$ non-zero elements), the assumptions of Theorem \ref{thm:2k_dft} are satisfied for almost all $\bm z$ with $\gamma$-shifted harmonics as samples (since almost all $\gamma$ satisfy $(\theta_k^+)^n \neq e^{-i\gamma}$, $0 \leq k \leq s-1$).

From Theorem \ref{thm:2k_dft} it is also possible to obtain $\bm g^+$.
\begin{corollary}
\label{corr:unique_g_2s}
Given the realizations $\bm \theta^+$, $\bm g^+$ of $\bm \theta$, $\bm g$, respectively. For almost all vectors $\bm z$ that consist of shifted harmonics, the unique solution to (\ref{eq:phase_msrmt_vandermonde}) is $\bm \theta = \bm \theta^+$, $\bm g = \bm g^+$, and it can be found with $O(\mathrm{poly}(s))$ complexity.
\end{corollary}
\begin{proof}
From Theorem \ref{thm:2k_dft} one obtains $\bm \theta^+$ and $q(z) = e^{i\gamma}\hat{u}(z) + \tilde{u}(z) = c(e^{i\gamma}\hat{u}^+(z) + \tilde{u}^+(z))$. It follows that $q((\theta_k^+)^{-1}) = cg_k^+t_{k}^+((\theta_k^+)^{-1})(e^{i\gamma}(\theta_k^+)^n - 1)$, $0 \leq k \leq s-1$, from which we get $\hat{\bm g} = c\bm g^+$ as
\begin{align*}
\hat{g}_k = cg_k^+ = \frac{q((\theta_k^+)^{-1})}{t_{k}^+((\theta_k^+)^{-1})(e^{i\gamma}(\theta_k^+)^n - 1)}, \quad 0 \leq k \leq s-1.
\end{align*}
It follows trivially from (\ref{eq:phase_msrmt_vandermonde}) that $$\hat{\bm y} = \bm V(\bm z)^T \bm V(\bm \theta^+)\hat{\bm g} = c\bm y^+$$ from which we easily obtain $c$ as $c = \hat{y}_k/y_k^+$ (for any $k$, $0 \leq k \leq m-1$, where $y_k^+ \neq 0$) and thereby $\bm g^+$ as $\bm g^+ = \hat{\bm g}/c$. It is clear that these operations for finding $\bm g^+$ are of $O(\mathrm{poly(s)})$ complexity.
\end{proof}
Clearly, another way to obtain $\bm g^+$ is by applying the pseudo-inverse: $\bm g^+ = \left(\bm V(\bm z)^T \bm V(\bm \theta^+)\right)^{\dagger}\bm y^+$.

To remove the constraint of shifted harmonic samples, we consider the more general equation system in (\ref{eq:phase_poly_samples_eqs}). First, we show the following.
\begin{theorem}
\label{thm:almost_gen_samples}
Given the assumptions in Theorem \ref{thm:2k_dft} and $z_j^n \neq e^{i\gamma}$, $v^+(z_j) \neq 0$ for $2s \leq j \leq 3s-1$ with $\gamma$ from Theorem \ref{thm:2k_dft}. Solving (\ref{eq:phase_poly_samples_eqs}) (or equivalently (\ref{eq:matrix_eqs})) gives that $v(z) = cv^+(z)$, $\hat{u}(z) = c\hat{u}^+(z)$ and $\tilde{u}^+(z) = c\tilde{u}(z)$ for some complex scalar $c$.
\end{theorem}
\begin{proof}
Inserting $y_j^+ = \frac{u^+(z_j)}{v^+(z_j)}$, $0 \leq j \leq 3s-1$, into (\ref{eq:phase_poly_samples_eqs}) and using $v^+(z_j) \neq 0$, we get 
\begin{align}
\label{lemma:id1}
(z_j^n\hat{u}^+(z_j) + \tilde{u}^+(z_j))v(z_j) = (z_j^n\hat{u}(z_j) + \tilde{u}(z_j))v^+(z_j), \quad 0 \leq j \leq 3s-1.
\end{align}
From Theorem \ref{thm:2k_dft}, it follows that after the $2s$ measurements $z_j = e^{i2\pi j/n}e^{i\gamma/n}$, $0 \leq j \leq 2s-1$, we obtain that $v(z) = cv^+(z)$ and $e^{i\gamma}\hat{u}(z) + \tilde{u}(z) = c(e^{i\gamma}\hat{u}^+(z) + \tilde{u}^+(z))$ hold for some complex scalar $c$. 
Inserting $v(z_j) = cv^+(z_j)$ into (\ref{lemma:id1}) for $j \geq 2s$ and dividing both sides of the equality with $v^+(z_j)$, we get 
\begin{align}
\label{lemma:id2}
\nonumber c(z_j^n\hat{u}^+(z_j) + \tilde{u}^+(z_j)) &= z_j^n\hat{u}(z_j) + \tilde{u}(z_j), \quad 2s \leq j \leq 3s-1 \\
\Rightarrow z_j^n(c\hat{u}^+(z_j) - \hat{u}(z_j)) &= -(c\tilde{u}^+(z_j) - \tilde{u}(z_j)), \quad 2s \leq j \leq 3s-1.
\end{align}
Rearranging the equality $e^{i\gamma}\hat{u}(z) + \tilde{u}(z) = c(e^{i\gamma}\hat{u}^+(z) + \tilde{u}^+(z))$ as $e^{i\gamma}(c\hat{u}^+(z) - \hat{u}(z)) = -(c\tilde{u}^+(z) - \tilde{u}(z))$ and inserting it into (\ref{lemma:id2}), it follows that
\begin{align}
\label{lemma:id3}
(c\hat{u}^+(z_j) - \hat{u}(z_j))(z_j^n - e^{i\gamma}) = 0, \quad 2s \leq j \leq 3s-1.
\end{align}
Utilizing the $s$ samples $z_j^n \neq e^{i\gamma}$, $2s \leq j \leq 3s-1$, and the fact that $\hat{u}^+(z)$, $\hat{u}(z)$ are polynomials of degree $s-1$, it follows that $c\hat{u}^+(z) = \hat{u}(z)$ holds for all $z$. Inserting this equality into $e^{i\gamma}(c\hat{u}^+(z) - \hat{u}(z)) = -(c\tilde{u}^+(z) - \tilde{u}(z))$ it also follows that $c\tilde{u}^+(z) = \tilde{u}(z)$.
\end{proof}
Clearly, if there are more than $3s$ samples, it is enough if $3s$ of them satisfy the assumptions of Theorem \ref{thm:almost_gen_samples} (while the rest can be arbitrary) to obtain $v(z) = cv^+(z)$, $\hat{u}(z) = c\hat{u}^+(z)$ and $\tilde{u}^+(z) = c\tilde{u}(z)$.

For a given $\bm z$ satisfying the assumptions of Theorem \ref{thm:almost_gen_samples}, almost all $\bm \theta^+$ and $\bm g^+$ will satisfy the assumptions of Theorem \ref{thm:almost_gen_samples}. Similarly, for any given $\bm \theta^+$ and $\bm g^+$ (with $s$ non-zero elements), almost all $\bm z$ with $2s$ samples that are $\gamma$-shifted harmonics and additional $s$ arbitrary samples (none being a $\gamma$-shifted harmonic) satisfy the assumptions of Theorem \ref{thm:almost_gen_samples}. Next, we relax the assumptions on $\bm z$.
\begin{theorem}
\label{thm:gen_samples_theta_g}
Assume that $n \geq 2s$ and $m = 3s$. Given any $\bm \theta^+ \in \mathbb{C}^s$ and $\bm g^+ \in \mathbb{C}^s$ (with $s$ non-zero elements), solving (\ref{eq:phase_poly_samples_eqs}) (or equivalently (\ref{eq:matrix_eqs})) gives that $v(z) = cv^+(z)$, $\hat{u}(z) = c\hat{u}^+(z)$ and $\tilde{u}^+(z) = c\tilde{u}(z)$ (for some complex scalar $c$) for almost all $\bm z \in \mathbb{C}^m$. Similarly, given almost any $\bm z \in \mathbb{C}^m$,  solving (\ref{eq:phase_poly_samples_eqs}) (or equivalently (\ref{eq:matrix_eqs})) gives that $v(z) = cv^+(z)$, $\hat{u}(z) = c\hat{u}^+(z)$ and $\tilde{u}^+(z) = c\tilde{u}(z)$ for almost all $\bm \theta^+ \in \mathbb{C}^s$ and $\bm g^+ \in \mathbb{C}^s$ (with $s$ non-zero elements).
\end{theorem}
\begin{proof} 
Herein, we will denote 
$\bm A(\bm z)$ in (\ref{eq:matrix_eqs}) as $\bm A(\bm z, \bm \theta^+, \bm g^+)$ to reflects its dependency on $\bm \theta^+, \bm g^+$ as well.
A one-dimensional null space of $\bm A(\bm z, \bm \theta^+, \bm g^+)$ implies that there exists a $3s \times 3s$ sub-matrix in $\bm A(\bm z, \bm \theta^+, \bm g^+)$ of full rank $3s$; this sub-matrix is obtained by removing a column from $\bm A(\bm z, \bm \theta^+, \bm g^+)$. Let $d_k(\bm z, \bm \theta^+, \bm g^+)$ denote the determinant of the $3s \times 3s$ sub-matrix when removing column $k$, $1 \leq k \leq 3s + 1$, from $\bm A(\bm z, \bm \theta^+, \bm g^+)$. Thus, the statement that $\bm A(\bm z, \bm \theta^+, \bm g^+)$ has a null space of at least two dimensions is equivalent to $d_k(\bm z, \bm \theta^+, \bm g^+) = 0, 1 \leq k \leq 3s + 1$. From $y^+(z_j) = u^+(z_j)/v^+(z_j)$, it follows that each $d_k(\bm z, \bm \theta^+, \bm g^+)$ is a \textit{multivariate rational function} in the variables $\bm z, \bm \theta^+, \bm g^+$.

First, we show that $\bm A(\bm z)$ has a one-dimensional null space for almost all $\bm z$ given a certain realization $\bm \theta^+ = \hat{\bm \theta}^+$ and $\bm g^+ = \hat{\bm g}^+$ (with $s$ non-zero elements). Theorem \ref{thm:almost_gen_samples} shows that for almost all $\bm z$ with $2s$ samples that are $\gamma$-shifted harmonics and $s$ samples that are not $\gamma$-shifted harmonics, $\bm A(\bm z, \hat{\bm \theta}^+, \hat{\bm g}^+)$ has a one dimensional null space; let $\bm z = \hat{\bm z}$ be one such vector. Hence, $d_j(\hat{\bm z}, \hat{\bm \theta}^+, \hat{\bm g}^+) \neq 0$ for some $j$, $1 \leq j \leq 3s + 1$, implying that $d_j(\bm z, \hat{\bm \theta}^+, \hat{\bm g}^+)$ is not zero for all $\bm z$. Thus, the zero set $\mathcal{S}_j = \{\bm z | \bm z \in \mathbb{C}^{3s}, d_j(\bm z, \hat{\bm \theta}^+, \hat{\bm g}^+) = 0\}$ is either finite or a hypersurface of dimension $3s - 1$ in the ambient space $\mathbb{C}^{3s}$. In either case, $\mathcal{S}_j$ has measure zero in $\mathbb{C}^{3s}$, implying that for almost all $\bm z \in \mathbb{C}^{3s}$,  $d_j(\bm z, \hat{\bm \theta}^+, \hat{\bm g}^+) \neq 0$ and thus $\bm A(\bm z, \hat{\bm \theta}^+, \hat{\bm g}^+)$ has a one-dimensional null space. Hence, for any given $\bm \theta^+$ and $\bm g^+$ (with $s$ non-zero elements), $\bm A(\bm z, \bm \theta^+, \bm g^+)$ has a one-dimensional null space for almost all $\bm z \in \mathbb{C}^{3s}$.

We now show that for almost any given $\bm z$, $\bm A(\bm z, \bm \theta^+, \bm g^+)$ has a one-dimensional null space for almost all $\bm \theta^+$ and $\bm g^+$ (with $s$ non-zero elements). Since almost all $\bm z$ do not belong to the set $\mathcal{S}_j$ above, choose any such element $\tilde{\bm z} \not \in \mathcal{S}_j$. Hence, $d_j(\tilde{\bm z}, \hat{\bm \theta}^+, \hat{\bm g}^+) \neq 0$, implying that the multivariate rational function $d_j(\tilde{\bm z}, \bm \theta^+, \bm g^+)$ in the variables $\bm \theta^+, \bm g^+$ is not zero. Therefore, the set $\mathcal{T}_j = \{\bm \theta^+, \bm g^+ | \bm \theta^+ \in \mathbb{C}^{s}, \bm g^+ \in \mathbb{C}^{s}, d_j(\tilde{\bm z}, \bm \theta^+, \bm g^+) = 0\}$ is either finite or is a hypersurface of dimension $2s - 1$ in the ambient space $\mathbb{C}^{2s}$. In both cases, it has measure zero in $\mathbb{C}^{2s}$, implying that $\bm A(\tilde{\bm z}, \bm \theta^+, \bm g^+)$ has a one-dimensional null space for almost all $\bm \theta^+, \bm g^+$. 
\end{proof}
Obviously, Theorem \ref{thm:gen_samples_theta_g} holds also when $m \geq 8s - 3$, since it can be applied to any subset of $8s - 3$ samples. The theorem shows that we can take almost any fixed vector of samples $\bm z$ to get the polynomials $v^+(z), \hat{u}^+(z), \tilde{u}^+(z)$. Similarly, it also shows that $v^+(z), \hat{u}^+(z), \tilde{u}^+(z)$ can be obtained for any given $\bm \theta^+$ and $\bm g^+$ (with $s$ non-zero elements) by choosing almost any vector $\bm z \in \mathbb{C}^m$ of measurement samples. Theorem \ref{thm:gen_samples_theta_g} also gives the following version of Corollary \ref{corr:unique_g_2s}
\begin{corollary}
\label{corr:unique_g_3s}
Given the realizations $\bm \theta^+$, $\bm g^+$ of $\bm \theta$, $\bm g$, respectively. For almost all $\bm z \in \mathbb{C}^m$ with $m \geq 3s$, the unique solution to (\ref{eq:phase_msrmt_vandermonde}) is $\bm \theta = \bm \theta^+$, $\bm g = \bm g^+$ and it can be found with $O(\mathrm{poly}(s))$ complexity.
\end{corollary}
\begin{proof}
Theorem \ref{thm:gen_samples_theta_g} shows that for almost all $\bm z \in \mathbb{C}^m$ with $m \geq 3s$, one obtains $\bm \theta^+$ and $\hat{u}(z) = c\hat{u}^+(z)$. It follows that $\hat{u}((\theta_k^+)^{-1}) = cg_k^+t_{k}^+((\theta_k^+)^{-1})$, $0 \leq k \leq s-1$, from which we get $\hat{\bm g} = c\bm g^+$ as
\begin{align*}
\hat{g}_k = cg_k^+ = \frac{u((\theta_k^+)^{-1})}{t_{k}^+((\theta_k^+)^{-1})}, \quad 0 \leq k \leq s-1.
\end{align*}
It follows trivially from (\ref{eq:phase_msrmt_vandermonde}) that $$\hat{\bm y} = \bm V(\bm z)^T \bm V(\bm \theta^+)\hat{\bm g} = c\bm y^+$$ from which we easily obtain $c$ as $c = \hat{y}_k/y_k^+$ (for any $k$, $0 \leq k \leq m-1$, where $y_k^+ \neq 0$) and thereby $\bm g^+$ as $\bm g^+ = \hat{\bm g}/c$. It is clear that these operations for finding $\bm g^+$ are of $O(\mathrm{poly(s)})$ complexity.
\end{proof}
In the above results, we assumed that $g_k^+ \neq 0$ for $0 \leq k \leq s-1$; i.e., we assumed exact knowledge of the number of non-zero elements in $\bm g^+$. If $g_k^+ = 0$ for some $k$, $\bm A(\bm z)$ and $\bm B(\bm z)$ can have null spaces of more than one dimension. Assume that $S$, $S \leq s$, of the $g_0^+,\ldots,g_{s-1}^+$ are non-zero. One can now iterate through different dimensions of $\bm A(\bm z)$ and $\bm B(\bm z)$, reducing the value of $s$ in each iteration by 1 until a one-dimensional null space is encountered. Since for each such dimension of $\bm A(\bm z)$ and $\bm B(\bm z)$, $\bm w^+$ in (\ref{eq:wp_phase}) and (\ref{eq:wp_dft_phase}) belongs to their null space, respectively, the solution will not be missed and the iteration terminates after at most $s-S+1$ steps (when $s$ is reduced to $S$, giving the unique solution $c\bm w^+$).
Hence, we can formulate the following recovery algorithm  
\par\noindent\rule{\textwidth}{0.4pt}
\textbf{R1-Alg:}
\begin{enumerate}
\item[] \textit{Input}: $s$, $n \geq 2s$, $\bm y^+$, samples $\bm z \in \mathbb{C}^m$ with $m \geq 2s$.
\item[] \textit{Assumptions}: If $2s \leq m < 3s$, then $m \leq n$ and all samples in $\bm z$ are shifted harmonics. If $m \geq 3s$, $\bm z$ can be arbitrary. 
\item[] \textit{Output}: $\bm \theta^+$, $\bm g^+$.
\item If $2s \leq m < 3s$, let $\bm C = \bm B(\bm z)$ from (\ref{eq:matrix_dft_eqs}); otherwise, let $\bm C = \bm A(\bm z)$ from (\ref{eq:matrix_eqs}).
\item Keep reducing the value of $s$ (and thus the dimension of $\bm C$) by one until $\bm C$ has a one-dimensional null space (assume $s = S$ when this occurs for the first time) of the form $c[\bm v^+; \bm q]$ for some complex scalar $c$ and vector $\bm q$. 
\item With vector $c\bm v^+$ construct the polynomial $cv^+(z)$ which has roots $(\theta_j^+)^{-1}$, $0 \leq j \leq S-1$, giving us $\bm \theta^+$.
\item Compute $\bm g^+ = \left(\bm V(\bm z)^T\bm V(\bm\theta^+)\right)^{\dagger}\bm y^+$ or use the procedures outlined in the proofs of Corollary \ref{corr:unique_g_2s} and Corollary \ref{corr:unique_g_3s}. 
\item Output $\bm \theta^+, \bm g^+$.
\end{enumerate}
\par\noindent\rule{\textwidth}{0.4pt}
For $m = O(s)$, it is obvious that the complexity of all these operations is $O(\mathrm{poly}(s))$.
\subsection{Derivation of \textbf{R2}}
Choose $\bm z \in \mathbb{C}^m$ and $\bm \theta \in \mathbb{C}^n$ such that $\bm z$ and any subset of $s$ elements from $\bm \theta$ satisfy assumptions in Theorem \ref{thm:2k_dft} or Theorem \ref{thm:almost_gen_samples}. As seen from these theorems, it is easy to choose such $\bm z$ and $\bm \theta$. Construct $\bm A$ in \textbf{R2} as $\bm A = \bm V(\bm z)^T \bm V(\bm \theta)$. In \textbf{R2}, denote the (at most) $s$ non-zero elements in $\bm x$ as vector $\bm g$ and assume they are on positions $p_0,\ldots,p_{s-1}$. Let $\bm p = [p_0,\ldots,p_{s-1}]$ and $\bm \theta_{\bm p} = [\theta_{p_0},\ldots,\theta_{p_{s-1}}]$. Hence, it follows that $\bm A\bm x = \bm V(\bm z)^T \bm V(\bm \theta)\bm x = \bm V(\bm z)^T \bm V(\bm \theta_{\bm p})\bm g$. Thus, using \textbf{R1} we obtain $\bm \theta_{\bm p}$ (from which we infer the positions $\bm p$) and $\bm g$ from which we easily reconstruct $\bm x$. Essentially, what is done in \textbf{R2} is to first create a grid $\bm \theta = [\theta_0,\ldots,\theta_{n-1}]$ through which $\bm x$ is measured. Thereafter, by using \textbf{R1}, we are able to find the $s$ activated grid elements which gives us the support of $\bm x$ from which we easily recover its elements through a simple matrix inversion. Basically, we have \textit{encoded} the support of $\bm x$ onto the grid elements in $\bm\theta$ so that recovering the support becomes equivalent to recovering the active grid elements. Thus, an algorithm with $O(\mathrm{poly}(s))$ complexity that recovers $\bm x$ in \textbf{R2} can be summarized as
\par\noindent\rule{\textwidth}{0.4pt}
\textbf{R2-Alg:}
\begin{enumerate}
\item[] \textit{Input}: $s$, $n \geq 2s$, $\bm y$, $\bm z \in \mathbb{C}^m$ with $m \geq 2s$, $\bm \theta \in \mathbb{C}^n$, $\bm A = \bm V(\bm z)^T \bm V(\bm \theta)$. 
\item[] \textit{Assumptions}: If $2s \leq m < 3s$, then $m \leq n$ and all samples in $\bm z$ are shifted harmonics; furthermore, $\bm z$ and $\bm \theta$ are chosen such that $\bm z$ and each subset of $s$ samples from $\bm \theta$ satisfy Theorem \ref{thm:2k_dft}. If $m \geq 3s$, $\bm z$ and $\bm \theta$ are chosen such that some $3s$ samples from $\bm z$ and each subset of $s$ samples from $\bm \theta$ satisfy Theorem \ref{thm:almost_gen_samples}.
\item[] \textit{Output}: $\bm x$.  
\item If $2s \leq m < 3s$, let $\bm C = \bm B(\bm z)$ from (\ref{eq:matrix_dft_eqs}); otherwise, let $\bm C = \bm A(\bm z)$ from (\ref{eq:matrix_eqs}). 
\item Keep reducing the value of $s$ (and thus the dimension of $\bm C$) by one until $\bm C$ has a one-dimensional null space (assume $s = S$ when this occurs for the first time) of the form $c[\bm v^+; \bm q]$ for some complex scalar $c$ and vector $\bm q$.
\item From vector $c\bm v^+$ construct the polynomial $cv^+(z)$ which has roots $r_0,\ldots,r_{S-1}$. 
\item Compute $$p_j = \arg \min_k |r_j - \theta_k^{-1}|, \quad 0 \leq j \leq S-1$$ and let $\bm \theta_{\bm p} = [\theta_{p_0},\ldots,\theta_{p_{S-1}}]$.
\item Compute $\bm g = \left(\bm V(\bm z)^T\bm V(\bm \theta_{\bm p})\right)^{\dagger}\bm y$ or use the procedures outlined in the proofs of Corollary \ref{corr:unique_g_2s} and Corollary \ref{corr:unique_g_3s}. 
\item Let $\bm x = \bm 0_n$ and set $x_{p_k} = g_k$, $0 \leq k \leq S-1$. Output $\bm x$.
\end{enumerate}
\par\noindent\rule{\textwidth}{0.4pt}
Since \textbf{R2-Alg} is essentially the same as \textbf{R1-Alg}, the complexity of \textbf{R2-Alg} is also $O(\mathrm{poly}(s))$.
\section{Phase-less Measurements}
\label{sec:phaseless_msrmt}
In this section, we look at phase-less measurements. Related to result \textbf{R2} is
\par\noindent\rule{\textwidth}{0.4pt}
\textbf{R3}: Assume $m \geq 4s$, $n \geq 4s - 1$ and construct $\bm A = [\bm V(\bm z)^T \bm V(\bm \theta);\bm a^T]$ where $\bm z \in \mathbb{T}^m, \bm \theta \in \mathbb{T}^n$ and $\bm a \in \mathbb{C}^n$ is some random vector. For infinitely many choices of $\bm z, \bm \theta$, almost all $\bm x \in \mathbb{C}_{\leq s}^n$ can be exactly recovered from
\begin{align}
\label{eq:phaseless_msrmt_construct}
\bm y = \left|\bm A \bm x\right|^2.
\end{align}
Furthermore, there exists a recovery algorithm such that
\begin{itemize}
\item If $4s \leq m < \min \{n+1, 8s-2\}$, the recovery complexity is $O(\mathrm{poly}(s)2^s)$.
\item If $m \geq 8s - 2$, the recovery complexity is $O(\mathrm{poly}(s))$.
\end{itemize}
\par\noindent\rule{\textwidth}{0.4pt}
This result provides a solution to problem 3.5 in \cite{xu} for \textit{almost all} $\bm x \in \mathbb{C}_{\leq s}^n$. Moreover, the result is significantly better than asked for by problem 3.5 in \cite{xu} since \textbf{R3} achieves $O(\mathrm{poly}(s))$ recovery complexity with $O(s)$ measurements. To obtain \textbf{R3}, we begin by considering the phase-less equivalent of \textbf{R1} with a Fourier matrix
\par\noindent\rule{\textwidth}{0.4pt}
\textbf{R4}: Assume an $n$-dimensional vector $\bm x = \bm F(\bm \theta)\bm g$, where $\bm \theta \in \mathbb{T}^s$, $\bm g \in \mathbb{C}^s$ and $n \geq 4s - 1$. Vector $\bm x$ is being measured with $m$ phase-less linear measurements of the form 
\begin{align}
\label{eq:phaseless_msrmt_fourier}
\bm y = \left|\bm F(\bm z)^T \bm x\right|^2 = \left|\bm F(\bm z)^T \bm F(\bm \theta)\bm g\right|^2
\end{align}
for some $\bm z \in \mathbb{T}^m$.
We have the following results
\begin{itemize}
\item If $4s - 1 \leq m < \min \{n+1, 8s-3\}$ and samples in $\bm z$ are shifted harmonics, $\bm \theta$ and $|\bm g|$ (up to a global real positive scalar) can be exactly recovered with $O(\mathrm{poly}(s))$ complexity and there are $2^{s-1}$ different $\bm g$ satisfying (\ref{eq:phaseless_msrmt_fourier}).
\item If $m \geq 8s-3$ and $\bm z$ is arbitrary, $\bm \theta$ and $|\bm g|$ (up to a global real positive scalar) can be exactly recovered with $O(\mathrm{poly}(s))$ complexity and there are only two different $\bm g$ satisfying (\ref{eq:phaseless_msrmt_fourier}).
\end{itemize}
\par\noindent\rule{\textwidth}{0.4pt}
To derive \textbf{R4} we will start (as with \textbf{R1}) by first recovering $\bm \theta$ from (\ref{eq:phaseless_msrmt_fourier}). However, in contrast to \textbf{R1}, we can no longer recover $\bm g$ with a simple matrix inversion due to the non-linear magnitude operator. After support recovery $\bm \theta$, we are (at first glance) essentially back to a classical phase retrieval problem. Below we will present an approach, similar in nature to the approach in Section \ref{sec:phase_measurements}, from which certain Laurent polynomials are solved for. It turns out that there is enough information about $\bm g$ embedded in these Laurent polynomials for recovering $\bm g$. 

\textbf{R4} shows that there is an ambiguity when measuring as in (\ref{eq:phaseless_msrmt_fourier}) which cannot be resolved by simply taking more measurements (increasing $m$). In order to resolve the ambiguity of which $\bm g$ is the correct one, one typically utilizes side information about the signal (e.g. knowledge of one element in $\bm g$ or $\bm x$) as explained in \cite{phase_ret_fourier} or take one additional measurement with some random measurement vector $\bm a$. Therefore, we have the following reformulation of \textbf{R4}
\par\noindent\rule{\textwidth}{0.4pt}
\textbf{R5}: Assume an $n$-dimensional vector $\bm x = \bm F(\bm \theta)\bm g$, where $\bm \theta \in \mathbb{T}^s$, $\bm g \in \mathbb{C}^s$ and $n \geq 4s - 1$. Vector $\bm x$ is being observed through $m+1$ phase-less linear measurements of the form 
\begin{align}
\label{eq:phase_less_msrmt_vandermonde}
\nonumber \bm y = \left|\bm F(\bm z)^T \bm x\right|^2 = \left|\bm F(\bm z)^T \bm F(\bm \theta)\bm g\right|^2 \\
y_m = |\bm a^T\bm x|^2 =  \left|\bm a^T \bm F(\bm \theta)\bm g\right|^2
\end{align}
for some $\bm z \in \mathbb{T}^m$ where $\bm a \in \mathbb{C}^n$ is a random vector. We have the following results
\begin{itemize}
\item If $4s - 1 \leq m < \min \{n+1, 8s-3\}$ and samples in $\bm z$ are shifted harmonics, $\bm \theta$ and $\bm g$ can be recovered exactly with $O(\mathrm{poly}(s)2^s)$ complexity.
\item If $m \geq 8s-3$ and for almost all $\bm z$, $\bm \theta$ and $\bm g$ can be recovered exactly with $O(\mathrm{poly}(s))$ complexity.
\end{itemize}
\par\noindent\rule{\textwidth}{0.4pt}

\subsection{Derivation of \textbf{R4} and \textbf{R5}}
Sample $y_j$ in \textbf{R4} is now seen as a \textit{phase-less} sample of the rational function $f(z)$ in (\ref{eq:f})
\begin{align}
\label{eq:y_j_rat_sum}
y_j = |f(z_j)|^2 = \left|\frac{z_j^n\hat{u}(z_j) + \tilde{u}(z_j)}{v(z_j)}\right|^2.
\end{align}
Let $\ol{(.)}$ denote the complex conjugate operator.
In order to deal with the non-linear magnitude operation, we express it as 
\begin{align}
\label{eq:y_rat_conj}
y(z) = \left(\frac{z^n\hat{u}(z) + \tilde{u}(z)}{v(z)}\right)\left(\frac{\ol{z^n\hat{u}(z) + \tilde{u}(z)}}{{\bar{v}(z)}}\right) = \left(\frac{z^n\hat{u}(z) + \tilde{u}(z)}{v(z)}\right)\left(\frac{z^{-n}\ol{\hat{u}(z)} + \ol{\tilde{u}(z)}}{{\ol{v(z)}}}\right).
\end{align}
Henceforth, it is assumed that $z \in \mathbb{T}$ and $\bm \theta \in \mathbb{T}^s$. Thus, $\bar{z} = 1/z$ and (\ref{eq:y_rat_conj}) becomes
\begin{align}
\label{eq:y_rat_conj_unit}
y(z) \stackrel{\triangle}{=} \frac{\hat{u}(z)\ol{\hat{u}(z)} + \tilde{u}(z)\ol{\tilde{u}(z)} + z^{-n}\ol{\hat{u}(z)}\tilde{u}(z) + z^n\hat{u}(z)\ol{\tilde{u}(z)}}{v(z)\ol{v(z)}}.
\end{align}
Due to $\ol{z} = 1/z$, the number of unknowns is reduced to just one ($z$) and the numerator and denominator in (\ref{eq:y_rat_conj_unit}) become Laurent polynomials in $z \in \mathbb{T}$.
Defining the Laurent polynomials
\begin{align}
\label{eq:laurent_polys}
\nonumber L(z) &\stackrel{\triangle}{=} \hat{u}(z)\ol{\hat{u}(z)} + \tilde{u}(z)\ol{\tilde{u}(z)} = |\hat{u}(z)|^2 + |\tilde{u}(z)|^2 \\
\nonumber \tilde{L}(z) &\stackrel{\triangle}{=} \hat{u}(z)\ol{\tilde{u}(z)} \\
\hat{L}(z) &\stackrel{\triangle}{=} v(z)\ol{v(z)}
\end{align}
where $L(z)$ and $\tilde{L}(z)$ are of at most degree $s-1$ in $z$ and $1/z$, respectively, while $\hat{L}_{s}(z)$ is of at most degree $s$ in $z$ and $1/z$, we can express $y(z)$ as a ratio of Laurent polynomials
\begin{align}
\label{eq:y_laurent_polys}
y(z) = \frac{L(z) + z^n\tilde{L}(z) + z^{-n}\ol{\tilde{L}(z)}}{\hat{L}(z)}.
\end{align}
Two properties of the introduced Laurent polynomials, important for this work, are given next.

\begin{prop}
\label{prop:laurent_poly_equiv}
A Laurent polynomial $P(z)$ with smallest degree $-s$ and largest degree $t$ can be expressed, for $z \neq 0$, as $\frac{1}{z^s}Q(z)$ where $Q(z)$ is a complex polynomial of degree $t+s$; thus, $P(z)$ and $Q(z)$ have the same non-zero roots. 
\end{prop}
\begin{proof}
Write $P(z) = \sum_{k=-s}^t p_kz^k = \frac{1}{z^s}\sum_{k=0}^{s+t} p_{k-s}z^k = \frac{1}{z^s}Q(z)$ where $Q(z) = \sum_{k=0}^{s+t} p_{k-s}z^k$.
\end{proof}
%\begin{prop}
%\label{prop:laurent_poly_unit_circle}
%The equalities $L(z) = \hat{u}(z)\ol{\hat{u}}(z) + \tilde{u}(z)\ol{\tilde{u}}(z)$, $\tilde{L}(z) = \hat{u}(z)\ol{\tilde{u}}(z)$ and $\hat{L}(z) = v(z)\ol{v}(z)$ are only valid for $z \in \mathbb{T}$ - they do not necessarily hold outside $\mathbb{T}$. 
%\end{prop}
%\begin{proof}
%This follows because the Laurent polynomials defined in (\ref{eq:y_laurent_polys}) are computed from $\hat{u}(z), \tilde{u}(z), v(z)$ by using the relationship $z\bar{z} = 1$ (which effectively reduces their degree); thus, when $z \not\in \mathbb{T}$, those equalities are not guaranteed to hold.
%\end{proof}

\begin{prop} 
\label{prop:laurent_roots}
Given two complex polynomials $p(z)$ and $r(z)$ with non-zero roots $p_0, \ldots, p_{k-1}$ and $r_0,\ldots,r_{l-1}$, respectively. Compute the Laurent polynomial $Q(z) = p(z)\ol{r(z)}$ with the assumption that $z \in \mathbb{T}$. It follows that the roots of $Q(z)$ are $p_0,\ldots,p_{k-1}$ and $1/\ol{r_0},\ldots,1/\ol{r_{l-1}}$. 
\end{prop}
\begin{proof}
The assumption $z \in \mathbb{T}$ gives that $\ol{r(z)}$ can equivalently be expressed as $\ol{r(z)} = \prod_{j=0}^{l-1}\ol{(z-r_j)} = \prod_{j=0}^{l-1}\left(\frac{1}{z}-\ol{r_j}\right) = \frac{1}{z^l}\prod_{j=0}^{l-1}\left(1-z\ol{r_j}\right) = \frac{1}{z^l}\hat{r}(z)$ where $\hat{r}(z)$ has $1/\ol{r_0},\ldots,1/\ol{r_{l-1}}$ as roots. Hence, we can write $Q(z)$ as $Q(z) = \frac{1}{z^l}p(z)\hat{r}(z)$ which completes the proof.
\end{proof}

%\begin{prop} 
%\label{prop:laurent_poly_id}
%For the Laurent polynomial $Q(z)$ in Proposition \ref{prop:laurent_roots}, if $p(z)$ and $r(z)$ have their roots in $\mathbb{T}$, then $Q(z) = p(z)\ol{r(z)}$ holds for \textit{all} $z$. 
%\end{prop}
%\begin{proof}
%When $p(z), r(z)$ have roots on the unit circle and $z \in \mathbb{T}$, $Q(z)$ will be \textit{exactly} equal to the polynomial multiplication $p(z)\ol{r(z)}$. In this case, $1/\ol{r_j} = r_j$, $0 \leq j \leq l-1$, and thus the roots of $Q(z)$ equal the roots of $p(z)$ and $r(z)$.
%\end{proof}

Proposition \ref{prop:laurent_roots} shows that the roots of $\hat{L}(z)$, $1/\theta_k = \ol{\theta_k}$, $0 \leq k \leq s-1$, appear with multiplicity two. 

When $n > 2s-2$, the Laurent polynomial in the numerator of (\ref{eq:y_laurent_polys}) contains at most $3(2s-1) = 6s-3$ parameters which is independent of $n$. Furthermore, $\hat{L}(z)$ contains at most $2s+1$ parameters. Hence, in total, $y(z)$ is described by $6s-3 + 2s + 1 = 8s - 2$ parameters. Thus, as in \cite{twc2023}, we can hope that $8s - 3$ values of $z$ from $\mathbb{T}$ are enough for recovering the $s$ values in $\bm \theta$ that identify the support of $\bm x$ (up to a scalar, due to one less measurement than number of unknowns, which is irrelevant for support identification).

Let $\bm l, \tilde{\bm l}, \hat{\bm l}$ denote the coefficient vectors of the Laurent polynomials $L(z), \tilde{L}(z), \hat{L}(z)$ with degrees $s, s-1, s-1$ in $z$ and $1/z$, respectively. Assume that $y^+_j = f^+(z_j)$, $0\leq j\leq m-1$, are the observed measurements which are produced by Laurent polynomials $L^+(z), \tilde{L}^+(z), \hat{L}^+(z)$ corresponding to realizations $\bm g = \bm g^+$ and $\bm \theta = \bm \theta^+$. Consider the equation system 
\begin{align}
\label{eq:sampled_poly_eqs}
y^+_j\hat{L}(z_j) - L(z_j) - z_j^n\tilde{L}(z_j) - z_j^{-n}\ol{\tilde{L}(z_j)} = 0, \quad 0 \leq j \leq m-1
\end{align}
in the vectors $\bm l, \tilde{\bm l}, \hat{\bm l}$. Equation system (\ref{eq:sampled_poly_eqs}) corresponds to the following linear equation system
\begin{align}
\label{eq:matrix_eqs_phaseless}
\nonumber \bm G(\bm z)\bm w &= \bm 0 \\
\nonumber &\textnormal{where} \\
\nonumber \bm G(\bm z) &= \left[\begin{array}{cccccc} \bm B(\bm z) & \bm y^+ & \mathrm{fliplr}\{\ol{\bm B(\bm z)}\} & -\bm C(\bm z) & -\bm 1 & -\mathrm{fliplr}\{\ol{\bm C(\bm z)}\} \end{array}\right] \\
\nonumber \bm B(\bm z) &= \left[\begin{array}{cccc} y_0^+z_0^s & y_0^+z_0^{s-1} & \ldots & y_0^+z_0 \\  & \vdots & \vdots & \\
y^+_{m-1}z_{m-1}^s & y^+_{m-1}z_{m-1}^{s-1} & \ldots & y^+_{m-1}z_{m-1} \end{array}\right] \\
\bm C(\bm z) &= \left[\begin{array}{cccccccc} z_0^{n+s-1} & z_0^{n+s-2} & \ldots & z_0^{n-s+1} & z_0^{s-1} & z_0^{s-2} & \ldots & z_0 \\  & \vdots & & \vdots & & \vdots & \\
z_{m-1}^{n+s-1} & z_{m-1}^{n+s-2} & \ldots & z_{m-1}^{n-s+1} & z_{m-1}^{s-1} & z_{m-1}^{s-2} & \ldots & z_{m-1}
 \end{array}\right].
\end{align}
where the operator $\mathrm{fliplr}\{.\}$ flips a matrix vertically (i.e., flips its columns), $\bm G(\bm z)$ has dimensions $m\times (8s-2)$ and $\bm w = [\hat{\bm l}; \tilde{\bm l}; \bm l; \ol{\tilde{\bm l}}]$. Assume that $S$ values from $g_0^+,\ldots,g_{s-1}^+$ are non-zero. Since $\hat{L}(z) = \hat{L}^+(z), L(z) = L^+(z), \tilde{L}(z) = \tilde{L}^+(z)$ solves (\ref{eq:sampled_poly_eqs}), 
\begin{align}
\label{eq:w+}
\bm w^+ = [\bm 0_{s-S};\hat{\bm l}^+;\bm 0_{s-S}; \bm 0_{s-S}; \tilde{\bm l}^+; \bm 0_{s-S}; \bm 0_{s-S}; \bm l^+; \bm 0_{s-S}; \bm 0_{s-S}; \ol{\tilde{\bm l}^+}; \bm 0_{s-S}]
\end{align}
solves (\ref{eq:matrix_eqs_phaseless}) and vice versa. This implies that $\bm G(\bm z)$ cannot be full rank and therefore has a null space of at least one dimension. If $\bm G(\bm z)$ has a one-dimensional null space, meaning that (\ref{eq:matrix_eqs_phaseless}) has the solution $\bm w = c\bm w^+$ for any real positive scalar $c$, then (\ref{eq:sampled_poly_eqs}) has the unique solution $\hat{L}(z) = c\hat{L}^+(z), L(z) = cL^+(z), \tilde{L}(z) = c\tilde{L}^+(z)$ and vice versa. 

In the special case of $z_j = e^{i2\pi j/n}e^{i\gamma/n}$ and $m \leq n$, (\ref{eq:sampled_poly_eqs}) reduces to 
\begin{align}
\label{eq:sampled_dft_poly_eqs}
y^+_j\hat{L}(z_j) - L(z_j) - e^{i\gamma/n}\tilde{L}(z_j) - e^{-i\gamma/n}\ol{\tilde{L}(z_j)} = 0, \quad  z_j = e^{i2\pi j/n}e^{i\gamma/n}, 0 \leq j \leq m-1
\end{align}
which corresponds to the following linear equation system
\begin{align}
\label{eq:matrix_dft_eqs_phaseless}
\nonumber \tilde{\bm G}(\bm z)\bm w &= \bm 0 \\
\nonumber &\textnormal{where} \\
\nonumber \tilde{\bm G}(\bm z) &= \left[\begin{array}{cccccc} \bm B(\bm z) & \bm y^+ & \mathrm{fliplr}\{\ol{\bm B(\bm z)}\} & -\tilde{\bm C}(\bm z) & -\bm 1 & -\mathrm{fliplr}\{\ol{\tilde{\bm C}(\bm z)}\} \end{array}\right] \\
\nonumber \bm B(\bm z) &= \left[\begin{array}{cccc} y_0^+z_0^s & y_0^+z_0^{s-1} & \ldots & y_0^+z_0 \\  & \vdots & \vdots & \\
y^+_{m-1}z_{m-1}^s & y^+_{m-1}z_{m-1}^{s-1} & \ldots & y^+_{m-1}z_{m-1} \end{array}\right] \\
\tilde{\bm C}(\bm z) &= \left[\begin{array}{cccc} z_0^{s-1} & z_0^{s-2} & \ldots & z_0 \\  & & \vdots & \\
z_{m-1}^{s-1} & z_{m-1}^{s-2} & \ldots & z_{m-1}
 \end{array}\right].
\end{align}
where $\tilde{\bm G}(\bm z)$ has dimensions $m \times 4s$ and $\bm w = [\hat{\bm l}; \bm l + e^{i\gamma/n}\tilde{\bm l} + e^{-i\gamma/n}\ol{\tilde{\bm l}}]$.
As with (\ref{eq:sampled_poly_eqs}) and (\ref{eq:matrix_eqs_phaseless}), a one-dimensional null space of $\tilde{\bm G}(\bm z)$ gives all solutions to (\ref{eq:matrix_dft_eqs_phaseless}) as $\bm w = c\bm w^+$, where 
\begin{align}
\label{eq:w+_dft}
\bm w^+ = [\bm 0_{s-S}; \hat{\bm l}; \bm 0_{s-S}; \bm 0_{s-S}; \bm l + e^{i\gamma/n}\tilde{\bm l} + e^{-i\gamma/n}\ol{\tilde{\bm l}}; \bm 0_{s-S}],
\end{align} for any real positive scalar $c$, which is equivalent to the solution $\hat{L}(z) = c\hat{L}^+(z)$, $L(z) + e^{i\gamma/n}\tilde{L}(z) + e^{-i\gamma/n}\ol{\tilde{L}(z)} = c\left(L^+(z) + e^{i\gamma/n}\tilde{L}^+(z) + e^{-i\gamma/n}\ol{\tilde{L}(z)^+}\right)$ of (\ref{eq:sampled_dft_poly_eqs}) and vice versa.

Before proceeding to solve (\ref{eq:sampled_poly_eqs}) and (\ref{eq:sampled_dft_poly_eqs}), we show that from each solution to these equations we can construct another solution to the equations. 
\begin{theorem}
\label{thm:dual_sol}
Assume that $\theta_i \in \mathbb{T}$, $0 \leq i \leq s-1$, are given. If $\bm g^a$ is a solution to (\ref{eq:sampled_poly_eqs}) and (\ref{eq:sampled_dft_poly_eqs}), then $\bm g^b$ with $g_l^b = \ol{g_l^a}\theta_l^{-n}\left(\prod_{\substack{i=0 \\ i\neq l}}^{s-1}\ol{\theta_i}\right)$, $0 \leq l \leq s-1$, is also a solution to (\ref{eq:sampled_poly_eqs}) and (\ref{eq:sampled_dft_poly_eqs}).
\end{theorem}
\begin{proof}
With $\bm g^a$ as a solution to (\ref{eq:sampled_poly_eqs}) and (\ref{eq:sampled_dft_poly_eqs}), the coefficients of the corresponding $\hat{u}(z), \tilde{u}(z)$ (denoted as $\hat{u}^a(z), \tilde{u}^a(z)$) are $g_k = g_k^a$. From the expressions of $L(z), \tilde{L}(z)$ in (\ref{eq:laurent_polys}), it follows that if we let $\hat{u}^b(z) = \ol{\tilde{u}^a(z)}$ and $\tilde{u}^b(z) = \ol{\hat{u}^a(z)}$, then these polynomials also satisfy $L(z) = |\hat{u}^b(z)|^2 + |\tilde{u}^b(z)|^2$ and $\tilde{L}(z) = \hat{u}^b(z)\ol{\tilde{u}^b(z)}$. Since $\hat{L}(z)$ is given and due to the assumption that $\theta_i$, $0 \leq i \leq s-1$, are given, $\hat{u}^b(z)$ and $\tilde{u}^b(z)$ are also a solution to (\ref{eq:sampled_poly_eqs}) and (\ref{eq:sampled_dft_poly_eqs}). Next, we find the expressions for $\hat{u}^b(z), \tilde{u}^b(z)$ from which we infer $g_k^b$. The definitions in (\ref{eq:u_exp}) give
\begin{align}
\nonumber \hat{u}^b(z) &= -\sum_{l=0}^{s-1}\ol{g^a_lt_{l}(z)} =
-\sum_{l=0}^{s-1}\ol{g^a_l}\prod_{\substack{i=0 \\ i\neq l}}^{s-1}(\bar{z}\ol{\theta_i}-1) \\
&= -\sum_{l=0}^{s-1}\ol{g^a_l}\prod_{\substack{i=0 \\ i\neq l}}^{s-1}\left(\frac{1}{z\theta_i}-1\right) \label{eqeq1} \\
&= -\frac{(-1)^{s-1}}{z^{s-1}}\sum_{l=0}^{s-1}\left(\ol{g^a_l}\prod_{\substack{i=0 \\ i\neq l}}^{s-1}\ol{\theta_i}\right)\prod_{\substack{i=0 \\ i\neq l}}^{s-1}\left(z\theta_i - 1\right) \label{eqeq2} \\
\nonumber &= -\frac{(-1)^{s-1}}{z^{s-1}}\sum_{l=0}^{s-1}\left(\ol{g^a_l}\theta_l^{-n}\prod_{\substack{i=0 \\ i\neq l}}^{s-1}\ol{\theta_i}\right)\theta_l^nt_{l}(z).
\end{align}
Equality (\ref{eqeq1}) holds because $z, \theta_i \in \mathbb{T}$ while (\ref{eqeq2}) follows from extracting the term $\frac{-1}{z\theta_i}$ from the product in (\ref{eqeq1}) and using the fact that $\frac{1}{\theta_i} = \ol{\theta_i}$. Similarly, from (\ref{eq:u_exp}), we have 
\begin{align*}
\nonumber \tilde{u}^b(z) &= \sum_{l=0}^{s-1}\ol{g^a_l}\theta_l^{-n}\ol{t_{l}(z)}
= \sum_{l=0}^{s-1}\ol{g^a_l}\theta_l^{-n}\prod_{\substack{i=0 \\ i\neq l}}^{s-1}(\bar{z}\ol{\theta_i}-1) \\
&= \sum_{l=0}^{s-1}\ol{g^a_l}\theta_l^{-n}\prod_{\substack{i=0 \\ i\neq l}}^{s-1}\left(\frac{1}{z\theta_i}-1\right)\\
\nonumber &= \frac{(-1)^{s-1}}{z^{s-1}}\sum_{l=0}^{s-1}\left(\ol{g^a_l}\theta_l^{-n}\prod_{\substack{i=0 \\ i\neq l}}^{s-1}\ol{\theta_i}\right)\prod_{\substack{i=0 \\ i\neq l}}^{s-1}\left(z\theta_i - 1\right) \\
&= \frac{(-1)^{s-1}}{z^{s-1}}\sum_{l=0}^{s-1}\left(\ol{g^a_l}\theta_l^{-n}\prod_{\substack{i=0 \\ i\neq l}}^{s-1}\ol{\theta_i}\right)t_{l}(z).
\end{align*}
Hence, $\hat{u}^b(z)$ and $\tilde{u}^b(z)$ are of the same form as $\frac{(-1)^{s-1}}{z^{s-1}}\hat{u}^a(z)$ and $\frac{(-1)^{s-1}}{z^{s-1}}\tilde{u}^a(z)$, respectively, from where it is also seen that $\hat{u}^b(z)$ and $\tilde{u}^b(z)$ produce the same Laurent polynomials as $\hat{u}^a(z), \tilde{u}^a(z)$ when $z \in \mathbb{T}$.
\end{proof}
From the expression of $\bm g^b$ in Theorem \ref{thm:dual_sol}, we get
\begin{align}
\label{eq:dual_dual_rel}
\ol{g_l^b}\theta_l^{-n}\left(\prod_{\substack{i=0 \\ i\neq l}}^{s-1}\bar{\theta}_i\right) = g_l^a\theta_l^n\theta_l^{-n}\left(\prod_{\substack{i=0 \\ i\neq l}}^{s-1}\theta_i\right)\left(\prod_{\substack{i=0 \\ i\neq l}}^{s-1}\ol{\theta_i}\right) = g_l^a
\end{align}
where the second equality holds due to $\theta_i \in \mathbb{T}$, $0 \leq i \leq s-1$. Hence, $\bm g^a$ can be obtained from $\bm g^b$ with the same transform from Theorem \ref{thm:dual_sol} that produces $\bm g^b$ from $\bm g^a$. Therefore, we will refer to $\bm g^b$ as being the \textit{dual} of $\bm g^a$ and vice versa. Clearly, the solutions to (\ref{eq:sampled_poly_eqs}) and (\ref{eq:sampled_dft_poly_eqs}) can be paired as $(\bm g^a, \bm g^b)$, with $\bm g^a$ and $\bm g^b$ being duals of each other. Inspired by the observations in \cite{eldar}, one can suspect that the Z-transforms of the corresponding solutions $\bm x^a = \bm F(\bm \theta)\bm g^a$ and $\bm x^b = \bm F(\bm \theta)\bm g^b$ are related. Indeed, there is a relation and it is revealed by the following corollary
\begin{corollary}
\label{corr:dual_sol}
Let $\bm r^a$ be the $n-1$ roots of $X^a(z)$, where $X^a(z) = x_0^a + x_1^az + \ldots x_{n-1}^az^{n-1}$ is the Z-transform of $\bm x^a = \bm F(\bm \theta)\bm g^a$ with $\bm g^a$ from Theorem \ref{thm:dual_sol}; in the same way we define $\bm r^b$ with $\bm g^b$ from Theorem \ref{thm:dual_sol}. It holds that  $\bm r^b = 1/\ol{\bm r^a}$ (up to a permutation), where the division is applied element-wise.  
\end{corollary}
\begin{proof}
From (\ref{eq:f}) it follows that $f(z) = X(z) = x_0 + x_1z + \ldots x_{n-1}z^{n-1}$ equals the Z-transform of $\bm x = \bm F(\bm \theta)\bm g$. The proof of Theorem \ref{thm:dual_sol} shows that $\tilde{u}^b(z) = \ol{\hat{u}^a(z)}$ and $\hat{u}^b(z) = \ol{\tilde{u}^a(z)}$, implying that for $\bm z \in \mathbb{T}$,
\begin{align}
\label{corr_eq:x_b_trans}
X^b(z) &= \frac{z^n\hat{u}^b(z) + \tilde{u}^b(z)}{v(z)} = \frac{z^n\ol{\tilde{u}^a(z)} + \ol{\hat{u}^a(z)}}{v(z)} = \frac{z^n\left(\ol{\tilde{u}^a(z)} + z^{-n}\ol{\hat{u}^a(z)}\right)}{v(z)}.
\end{align}
The roots $\bm r^b$ of $X^b(z)$ are roots of $z^n\hat{u}^b(z) + \tilde{u}^b(z)$ and the roots $\bm r^a$ of $X^a(z)$ are roots of $z^n\hat{u}^a(z) + \tilde{u}^a(z)$.
Equation (\ref{corr_eq:x_b_trans}) shows that $z^n\hat{u}^b(z) + \tilde{u}^b(z) = z^n\ol{z^n\hat{u}^a(z) + \tilde{u}^a(z)}$ for $z \in \mathbb{T}$ and by using
Proposition \ref{prop:laurent_roots} it follows that the non-zero roots of $X^b(z)$ equal $1/\ol{\bm r^a}$.
\end{proof}
Hence, the roots of the Z-transforms $X^a(z), X^b(z)$ are reflections of each other with respect to $\mathbb T$, which are regarded as "trivial" solutions in \cite{eldar}. 

We start by providing a solution to (\ref{eq:matrix_dft_eqs_phaseless}). As described, we first investigate the minimum number of samples needed to recover the support $\bm \theta^+$.
\begin{theorem}
\label{thm:dft_msrmt_support}
Assume that $n \geq 4s - 1$, $z_j^n = e^{i\gamma}$ for $0 \leq j \leq 4s - 2$,  $g_k^+ \neq 0$, $\theta_k^+ \in \mathbb{T}$ and $(\theta_k^+)^n \neq e^{-i\gamma}$ for $0 \leq k \leq s-1$. Solving (\ref{eq:sampled_dft_poly_eqs}) (or equivalently (\ref{eq:matrix_dft_eqs_phaseless})) results in $\hat{L}(z) = c\hat{L}^+(z), L(z) + z^n\tilde{L}(z) + z^{-n}\ol{\tilde{L}}(z) = c\left(L^+(z) + z^n\tilde{L}^+(z) + z^{-n}\ol{\tilde{L}^+}(z)\right)$ for some real positive scalar $c$.
\end{theorem}
\begin{proof}
Inserting $y^+(z_j) = \left(L^+(z_j) + z_j^n\tilde{L}^+(z_j) + z_j^{-n}\ol{\tilde{L}^+(z_j)}\right)/\hat{L}^+(z_j)$ into (\ref{eq:sampled_poly_eqs}) and multiplying both sides of the equation with $\hat{L}^+(z_j)$ results in the equation system
\begin{align}
\label{eq:laurent_poly_sys}
\left(L^+(z_j) + z_j^n\tilde{L}^+(z_j) + z_j^{-n}\ol{\tilde{L}^+(z_j)}\right)\hat{L}(z_j) - \left(L(z_j) + z_j^n\tilde{L}(z_j) + z_j^{-n}\ol{\tilde{L}(z_j)}\right)\hat{L}^+(z_j) &= 0
\end{align}
which holds for all $z_j$ such that $\hat{L}^+(z_j) \neq 0$. Since $z_j^n = e^{i\gamma}$ and $(\theta_k^+)^n \neq e^{-i\gamma}$ for $0 \leq j \leq 4s-2$, $0 \leq k \leq s-1$, it follows that $\hat{L}^+(z_j) \neq 0$ for $0 \leq j \leq 4s - 2$.
Using the fact that $z_j^n = e^{i\gamma}$ for $0 \leq j \leq 4s-2$, (\ref{eq:laurent_poly_sys}) gives
\begin{align}
\label{eq:laurent_poly_sys_dft}
\nonumber &\left(L^+(z_j) + e^{i\gamma}\tilde{L}^+(z_j) + e^{-i\gamma}\ol{\tilde{L}^+(z_j)}\right)\hat{L}(z_j) \\
&- (L(z_j) + e^{i\gamma}\tilde{L}(z_j) + e^{-i\gamma}\ol{\tilde{L}(z_j)})\hat{L}^+(z_j) = 0, 
 \quad 0 \leq j \leq 4s - 2. 
\end{align}
Since $z_j \neq 0$, it follows from Proposition \ref{prop:laurent_poly_equiv} that equation (\ref{eq:laurent_poly_sys_dft}) is equivalent to a \textit{polynomial} equation in $z_j$ of degree $4s-2$. Applying the identity theorem for polynomials to (\ref{eq:laurent_poly_sys_dft}) gives that the relation 
\begin{align}
\label{eq:laurent_poly_rel}
\left(L^+(z) + e^{i\gamma}\tilde{L}^+(z) + e^{-i\gamma}\ol{\tilde{L}^+(z)}\right)\hat{L}(z) =
\left(L(z) + e^{i\gamma}\tilde{L}(z) + e^{-i\gamma}\ol{\tilde{L}(z)}\right)\hat{L}^+(z)
\end{align}
holds for \textit{all} $z$. Proposition \ref{prop:laurent_roots} shows that the roots of $\hat{L}^+(z)$ are $1/\theta_k^+$, $0 \leq k \leq s-1$, where each one of them has multiplicity two. These roots are not roots of $L^+(z) + e^{i\gamma}\tilde{L}^+(z) + e^{-i\gamma}\ol{\tilde{L}^+(z)}$. Assume the contrary. Noting that $L^+(z) + e^{i\gamma}\tilde{L}^+(z) + e^{-i\gamma}\ol{\tilde{L}^+(z)} = \left|e^{i\gamma}\hat{u}^+(z) + \tilde{u}^+(z)\right|^2$
implies that the roots of $\hat{L}^+(z)$ must be roots of $e^{i\gamma}\hat{u}^+(z) + \tilde{u}^+(z)$. It follows from (\ref{eq:u_exp}) that $$e^{i\gamma}\hat{u}^+\left(\frac{1}{\theta_k^+}\right) + \tilde{u}^+\left(\frac{1}{\theta_k^+}\right) = g_k^+t_{k}^+\left(\frac{1}{\theta_k^+}\right)\left(e^{i\gamma}(\theta_k^+)^n - 1\right).$$ Since $(\theta_k^+)^n \neq e^{-i\gamma}$ and $g_k \neq 0$, it follows that $t_{k}^+\left(\frac{1}{\theta_k^+}\right) \neq 0$ and thus $g_k^+t_{k}^+\left(\frac{1}{\theta_k^+}\right) \neq 0$, which implies that the roots of $\hat{L}^+(z)$ are not roots of $L^+(z) + e^{i\gamma}\tilde{L}^+(z) + e^{-i\gamma}\ol{\tilde{L}^+(z)}$, a contradiction. Hence, the roots of $\hat{L}^+(z)$ are roots of $\hat{L}(z)$ which implies that $\hat{L}(z) = c\hat{L}^+(z)$ for some real positive scalar $c$.

Inserting $\hat{L}(z) = c\hat{L}^+(z)$ into (\ref{eq:laurent_poly_rel}) and dividing out $\hat{L}^+(z)$, we conclude that $$c\left(L^+(z) + e^{i\gamma}\tilde{L}^+(z) + e^{-i\gamma}\ol{\tilde{L}^+(z)}\right) = L(z) + e^{i\gamma}\tilde{L}(z) + e^{-i\gamma}\ol{\tilde{L}(z)}.$$
\end{proof}
From the result in Theorem \ref{thm:dft_msrmt_support}, it follows that when $g_k^+ \neq 0$, $0 \leq k \leq s-1$, we can find $\hat{L}^+(z)$ (up to a real positive scalar $c$) after $4s - 1$ measurements with shifted harmonic samples from which we can recover the support $\bm \theta^+$ (below we will deal with the case when $S$, $S \leq s$, of $g_0^+,\ldots,g_{s-1}^+$ are non-zero). Furthermore, from these measurements, we also obtain $L^+(z) + e^{i\gamma}\tilde{L}^+(z) + e^{-i\gamma}\ol{\tilde{L}^+(z)} = \left|e^{i\gamma}\hat{u}^+(z) + \tilde{u}^+(z)\right|^2$ (up to a real positive scalar $c$). Given samples that are shifted harmonics, almost all $\bm \theta^+ \in \mathbb{T}^s$ and $\bm g^+ \in \mathbb{C}^s$ (with $s$ non-zero elements) satisfy the assumptions of Theorem \ref{thm:dft_msrmt_support}. Similarly, for any $\bm \theta^+ \in \mathbb{T}^s$ and $\bm g^+ \in \mathbb{C}^s$ (with $s$ non-zero elements), almost all $\gamma$-shifted harmonics satisfy the assumptions of Theorem \ref{thm:dft_msrmt_support}. 

From Theorem \ref{thm:dft_msrmt_support} one can find $|\bm g^+|$ up to a real positive scalar.
\begin{corollary}
\label{corr:unique_g_abs_4s}
Given the realizations $\bm \theta^+$, $\bm g^+$ of $\bm \theta$, $\bm g$, respectively. For almost all vectors $\bm z$ that consist of at least $4s-1$ shifted harmonics, $\bm \theta^+$ and $|\bm g^+|$ can be found, up to a real positive scalar, with $O(\mathrm{poly}(s))$ complexity.
\end{corollary}
\begin{proof}
From Theorem \ref{thm:dft_msrmt_support} one obtains $\bm \theta^+$ and $$q(z) = L(z) + z^n\tilde{L}(z) + z^{-n}\ol{\tilde{L}(z)} = c\left(L^+(z) + z^n\tilde{L}^+(z) + z^{-n}\ol{\tilde{L}^+(z)}\right) = c|e^{i\gamma}\hat{u}^+(z) + \tilde{u}^+(z)|^2$$ for some unknown real positive scalar $c$.
Since $(\theta_k^+)^{-1} = \ol{\theta_k^+}$ and $$q(\ol{\theta_k^+}) = c\left|e^{i\gamma}\hat{u}^+(\ol{\theta_k^+}) + \tilde{u}^+(\ol{\theta_k^+})\right|^2 = c\left|g_k^+t_{k}^+\left(\ol{\theta_k^+}\right)\left(e^{i\gamma}(\theta_k^+)^n - 1\right)\right|^2,$$ it follows that $$c|g_k^+|^2 = \frac{q(\ol{\theta_k^+})}{\left|t_{k}^+\left(\ol{\theta_k^+}\right)\left(e^{i\gamma}(\theta_k^+)^n - 1\right)\right|^2}.$$ Clearly, the operations leading to $c|\bm g^+|^2$ are of $O(\mathrm{poly}(s))$ complexity, showing that $\sqrt{c}|\bm g^+|$ can be found with $O(\mathrm{poly}(s))$ complexity.
\end{proof}
A direct application of this proposition is that we can find the \textit{relative} magnitudes of the different frequency components with $4s - 1$ measurements and $O(\mathrm{poly}(s))$ complexity.

When it comes to recovering $\bm g^+$, the situation is trickier. Given $\left|e^{i\gamma}\hat{u}^+(z) + \tilde{u}^+(z)\right|^2$, up to a real positive scalar, Proposition \ref{prop:laurent_roots} implies that the roots of $\left|e^{i\gamma}\hat{u}^+(z) + \tilde{u}^+(z)\right|^2$ appear as conjugate pairs: if $r$ is a root of $\left|e^{i\gamma}\hat{u}^+(z) + \tilde{u}^+(z)\right|^2$ then so is $1/\ol{r}$ as well. Let $(r_0, 1/\ol{r_0}),\ldots,(r_{s-2}, 1/\ol{r_{s-2}})$ be the $s-1$ conjugate root pairs. As seen from (\ref{eq:u_exp}), $e^{i\gamma}\hat{u}^+(z) + \tilde{u}^+(z)$ is of the same form as $\tilde{u}^+(z)$ in (\ref{eq:u_exp}) but with coefficients $(e^{i\gamma}(\theta_k^+)^n - 1)g_k^+$, $0 \leq k \leq s-1$. Hence, for almost all $\bm g^+$, Lemma \ref{lemma:u_tilde_distinct_roots} implies that $(r_0, 1/\ol{r_0}),\ldots,(r_{s-2}, 1/\ol{r_{s-2}})$ are distinct.
Only one member from each pair can be a root of $e^{i\gamma}\hat{u}^+(z) + \tilde{u}^+(z)$. Since we do not know beforehand which member of each pair is the actual root of $e^{i\gamma}\hat{u}^+(z) + \tilde{u}^+(z)$,  we need to go through all the $2^{s-1}$ possibilities. Denote by $\bm q^k = [q_0^k,\ldots,q_{s-2}^k]$ the $k$:th possibility, meaning that $q_j^k$ is either $r_j$ or $1/\ol{r_j}$ for $0 \leq j \leq s-2$, $0 \leq k \leq 2^{s-1}-1$. For each $\bm q^k$, we get an equation system in $\bm g$
\begin{align}
\label{eq:root_eq_sys}
e^{i\gamma}\tilde{u}(q_j^k) + \hat{u}(q_j^k) = \sum_{\substack{l=0}}^{s-1}g_lt_{l}^+(q_j^k)(e^{i\gamma}(\theta_l^+)^n  - 1) = 0, \quad 0 \leq j \leq s-2.
\end{align}
The equation system in (\ref{eq:root_eq_sys}) with $s-1$ equations represents the zeros of a polynomial of degree $s-1$ and thus it has a unique solution $\bm g^k$ up to a complex scalar. Each $\bm g^k$, with proper normalization, is a solution to (\ref{eq:sampled_dft_poly_eqs}), i.e., $\bm y^+ = |\bm V^T(\bm z)\bm V(\bm \theta^+)\bm g^k|^2$. Therefore, after $4s-1$ measurements as in \textbf{R4} (with shifted harmonic samples) there are $2^{s-1}$ solutions to \textbf{R4}. To resolve the ambiguity of which $\bm g^k$ equals $\bm g^+$, we need an additional measurement - this is the purpose of the random vector $\bm a$ in \textbf{R5}. With probability 1, only one of $\bm g^k$, $0 \leq k \leq 2^{s-1}-1$, will satisfy $y_m^+ = |\bm a^T\bm V(\bm \theta^+)\bm g^k|^2$, resolving the ambiguity and thus giving us $\bm g^+$. 

In the derivations above, we assumed that $g_k^+ \neq 0$, $0 \leq k \leq s-1$; i.e., we assumed exact knowledge of the number of non-zero elements. If $g_k^+ = 0$ for some $k$, $\tilde{\bm G}(\bm z)$ in (\ref{eq:matrix_dft_eqs_phaseless}) can have a null space of more than one dimension. Assume that $S$, $S \leq s$, of the $g_0^+,\ldots,g_{s-1}^+$ are non-zero. One can now iterate through different realizations of $\tilde{\bm G}(\bm z)$, reducing the value of $s$ in each step by 1 until a one-dimensional null space is encountered. Since for each such realization of $\tilde{\bm G}(\bm z)$, $\bm w^+$ in (\ref{eq:w+_dft}) belongs to its null space, the solution will not be missed and the iteration terminates after at most $s-S+1$ steps (when $s$ is reduced to $S$, giving the unique solution $c\bm w^+$ as implied by the derivations above).

Thus, an algorithm for producing the $2^{s-1}$ solutions to \textbf{R4} and that gives the exponential recovery complexity in \textbf{R5} is
\par\noindent\rule{\textwidth}{0.5pt}
\textbf{R5-Alg-Harmonic:}
\begin{enumerate}[label=\arabic*.]
\item[] \textit{Input}: $s, n \geq 4s-1, \bm y^+$, a random vector $\bm a \in \mathbb{C}^n$ and samples $\bm z \in \mathbb{T}^m$ where $4s-1 \leq m < \min\{n+1,8s-3\}$.
\item[] \textit{Assumptions}: All samples in $\bm z$ are shifted harmonics.
\item[] \textit{Output}: $\bm \theta^+, e^{i\alpha}\bm g^+$ for some unknown $\alpha$
\item Construct the $m \times 4s$ matrix $\tilde{\bm G}(\bm z)$ in (\ref{eq:matrix_dft_eqs_phaseless}).
\item Keep reducing the value of $s$ (and thus the dimension of $\tilde{\bm G}(\bm z)$) by one until $\tilde{\bm G}(\bm z)$ has a one-dimensional null space (assume $s = S$ when this occurs for the fist time). Let $\bm w$ denote a vector from its one-dimensional null space.
\item Construct the Laurent polynomial $\hat{L}(z) = \sum_{k=0}^{2S}w_kz^{k-S}$. The conjugate of the unique roots of $\hat{L}(z)$ equal $\bm \theta^+$. $|\bm g^+|$ can now be computed (up to a global real positive scalar) from Corollary \ref{corr:unique_g_abs_4s}.
\item Construct the Laurent polynomial $Q(z) = \sum_{k=0}^{2S-2}w_{k+2S+1}z^{k-S+1}$.
\item The $2S-2$ roots of $Q(z)$ can be paired as $(r_0, 1/\ol{r_0}),\ldots,(r_{S-2},1/\ol{r_{S-2}})$. From each pair, choose one element. This choice can be done in $2^{S-1}$ different ways and let $\bm q^k$, $0 \leq k \leq 2^{S-1}-1$, denote the $k$:th choice.
\item For each $\bm q^k$, solve (\ref{eq:root_eq_sys}) with $s = S$ and let $\bm g^k$ denote the solution (up to a complex scalar). Normalize each $\bm g^k$ so that $\bm y^+ = |\bm V^T(\bm z)\bm V(\bm \theta^+)\bm g^k|^2$
\item Let $\bm g^+ = \arg \min_{\bm g \in \{\bm g^0,\ldots,\bm g^{2^{S-1}-1}\}} |y_m - |\bm a^T\bm V(\bm \theta^+)\bm g|^2|$. Output $\bm g^+$ and $\bm \theta^+$.
\end{enumerate}
\par\noindent\rule{\textwidth}{0.4pt}
From the theoretical results above, it follows that \textbf{R5-Alg-Harmonic} recovers almost all $\bm \theta^+$ and $\bm g^+$. Steps 1-6 produce all solutions to \textbf{R4} while step 7 finds the correct $\bm g^+$ and thereby gives the result in \textbf{R5}. The complexity of \textbf{R5-Alg-Harmonic} is $O(\mathrm{poly}(s)2^{s})$. Up until step 6, the total complexity is $O(\mathrm{poly}(s))$. Solving (\ref{eq:root_eq_sys}) has complexity $O(\mathrm{poly}(s))$ but it is repeated $2^{s-1}$ times in step 6. It is possible to utilize Theorem \ref{thm:dual_sol} to directly generate the dual solution for each $\bm g^k$ and thereby reduce the number of times (\ref{eq:root_eq_sys}) is solved by half, but this still results in complexity $O(\mathrm{poly}(s)2^{s})$. Fortunately, it turns out that by taking $4s-2$ additional measurements, the number of possible solutions can be reduced from $2^{s-1}$ to only two. 
\begin{theorem}
\label{thm:dft_msrmt_unique}
Given the assumptions in Theorem \ref{thm:dft_msrmt_support}, further assume that $z_j^n = e^{i\omega}$, $\hat{L}^+(z_j) \neq 0$ for $4s-1 \leq j \leq 6s-3$ and some $\omega$ satisfying $e^{i\gamma} \neq e^{i\omega}$ (with $\gamma$ from Theorem \ref{thm:dft_msrmt_support}); $z_j^n = e^{i\phi}$, $\hat{L}^+(z_j) \neq 0$ for $6s - 2 \leq j \leq 8s - 4$ and some $\phi$ satisfying $e^{i\phi} \neq e^{i\gamma}$ and $(e^{i\phi} - e^{i\gamma})(e^{-i\omega} - e^{-i\gamma}) \neq (e^{-i\phi} - e^{-i\gamma})$. Solving (\ref{eq:sampled_poly_eqs}) (or equivalently (\ref{eq:matrix_eqs_phaseless})) results in $\hat{L}(z) = c\hat{L}^+(z), L(z) = cL^+(z), \tilde{L}(z) = c\tilde{L}^+(z)$ for some real positive scalar $c$.
\end{theorem}
\begin{proof}
After the $4s - 1$ measurements $z_j$, $0 \leq j \leq 4s - 2$, that satisfy $z_j^n = e^{i\gamma}$ we conclude from Theorem \ref{thm:dft_msrmt_support} that $\hat{L}(z) = c\hat{L}^+(z)$. Inserting this relation into (\ref{eq:laurent_poly_rel}) and dividing out $\hat{L}^+(z)$, we get
\begin{align}
\label{eq:laurent_poly_dft_rel}
cL^+(z) - L(z) = e^{i\gamma}(\tilde{L}(z) - c\tilde{L}^+(z)) + e^{-i\gamma}(\ol{\tilde{L}(z)} - c\ol{\tilde{L}^+(z)}).
\end{align}
Further using $\hat{L}(z) = c\hat{L}^+(z)$ in (\ref{eq:laurent_poly_sys}) and dividing out $\hat{L}^+(z_j)$ gives that
\begin{align}
\label{eq:laurent_poly_sys_rel}
cL^+(z_j) - L(z_j) = z_j^n(\tilde{L}(z_j) - c\tilde{L}^+(z_j)) + z_j^{-n}(\ol{\tilde{L}(z_j)} - c\ol{\tilde{L}^+(z_j)})
\end{align}
holds for all $z_j$ for which $\hat{L}^+(z_j) \neq 0$.
Subtracting (\ref{eq:laurent_poly_dft_rel}), evaluated at $z = z_j$, from (\ref{eq:laurent_poly_sys_rel}) we get 
\begin{align}
(z_j^n - e^{i\gamma})(\tilde{L}(z_j) - c\tilde{L}^+(z_j)) + (z_j^{-n} - e^{-i\gamma})(\ol{\tilde{L}(z_j)} - c\ol{\tilde{L}^+(z_j)}) &= 0. \label{eq:lauren_poly_id_1}
\end{align}
Proposition \ref{prop:laurent_poly_equiv} shows that (\ref{eq:lauren_poly_id_1}) is equivalent to a polynomial equation in $z_j$ of degree $2s-2$. Hence, with the $2s - 1$ samples $z_j$, $4s-1 \leq j \leq 6s-3$, for which $z_j^n = e^{i\omega}$, (\ref{eq:lauren_poly_id_1}) becomes
\begin{align}
(e^{i\omega} - e^{i\gamma})(\tilde{L}(z_j) - c\tilde{L}^+(z_j)) + (e^{-i\omega} - e^{-i\gamma})(\ol{\tilde{L}(z_j)} - c\ol{\tilde{L}^+(z_j)}) &= 0, \quad 4s-1\leq j \leq 6s-3. \label{eq:lauren_poly_id_2}
\end{align}
Applying the identity theorem for polynomials to the equations in (\ref{eq:lauren_poly_id_2}) implies that
\begin{align}
(e^{i\omega} - e^{i\gamma})(\tilde{L}(z) - c\tilde{L}^+(z)) + (e^{-i\omega} - e^{-i\gamma})(\ol{\tilde{L}(z)} - c\ol{\tilde{L}^+(z)}) &= 0, \quad \forall z \label{eq:lauren_poly_id_all_z_1} 
\end{align}
Using the $2s - 1$ samples $z_j = e^{i2\pi j/n}e^{i\phi}$, $6s-2 \leq j \leq 8s-4$, in (\ref{eq:lauren_poly_id_1}), we obtain
\begin{align}
(e^{i\phi} - e^{i\gamma})(\tilde{L}(z_j) - c\tilde{L}^+(z_j)) + (e^{-i\phi} - e^{-i\gamma})(\ol{\tilde{L}(z_j)} - c\ol{\tilde{L}^+(z_j)}) &= 0, \quad 6s-2\leq j \leq 8s-4. \label{eq:lauren_poly_id_3}
\end{align}
Applying again the identity theorem for polynomials, the equations in (\ref{eq:lauren_poly_id_3}) imply that 
\begin{align}
(e^{i\phi} - e^{i\gamma})(\tilde{L}(z) - c\tilde{L}^+(z)) + (e^{-i\phi} - e^{-i\gamma})(\ol{\tilde{L}(z)} - c\ol{\tilde{L}^+(z)}) &= 0, \quad \forall z. \label{eq:lauren_poly_id_all_z_2} 
\end{align}
Multiplying (\ref{eq:lauren_poly_id_all_z_1}) with $\frac{(e^{-i\phi} - e^{-i\gamma})}{(e^{-i\omega} - e^{-i\gamma})}$, which is a finite number due to the assumptions on $\phi, \omega, \gamma$,  and subtracting it from (\ref{eq:lauren_poly_id_all_z_2}) we get
\begin{align}
\label{eq:laurent_poly_id_subtract}
\left((e^{i\phi} - e^{i\gamma}) - \frac{(e^{-i\phi} - e^{-i\gamma})}{(e^{-i\omega} - e^{-i\gamma})}\right)(\tilde{L}(z) - c\tilde{L}^+(z)) = 0, \quad \forall z.
\end{align}
The assumptions on $\phi, \omega, \gamma$ imply that the scalar in front of $\tilde{L}(z) - c\tilde{L}^+(z)$ in (\ref{eq:laurent_poly_id_subtract}) is non-zero. Thus, the identity theorem for polynomials implies that $\tilde{L}(z) = c\tilde{L}^+(z)$ must hold in (\ref{eq:laurent_poly_id_subtract}). Inserting this relation into (\ref{eq:laurent_poly_dft_rel}) we also obtain $L(z) = cL^+(z)$, finishing the proof.
\end{proof}
Clearly, Theorem \ref{thm:dft_msrmt_unique} holds even if $m \geq 8s - 3$ as long as a subset of $8s - 3$ measurements satisfies the assumptions of  Theorem \ref{thm:dft_msrmt_unique}. For any $\bm z$ consisting of shifted harmonics, with the different shifts specified in Theorem \ref{thm:dft_msrmt_unique}, almost $\bm \theta^+ \in \mathbb{T}^s$ and $\bm g^+$ satisfy the assumptions of Theorem \ref{thm:dft_msrmt_unique}. Similarly, for any given $\bm \theta^+ \in \mathbb{T}^s$ and $\bm g^+$ with no zero elements, almost all $\bm z$ consisting of shifted harmonics as specified in Theorem \ref{thm:dft_msrmt_unique} satisfy the assumptions of Theorem \ref{thm:dft_msrmt_unique}. Similar to Theorem \ref{thm:gen_samples_theta_g}, we can relax the assumptions on $\bm z$.
\begin{theorem}
\label{thm:gen_phaseless_samples_theta_g}
Assume that $n \geq 4s - 1$ and $m = 8s - 3$. Given any $\bm \theta^+ \in \mathbb{T}^s$ and any $\bm g^+ \in \mathbb{C}^s$ (with $s$ non-zero elements), solving (\ref{eq:sampled_poly_eqs}) (or equivalently (\ref{eq:matrix_eqs_phaseless})) results in $\hat{L}(z) = c\hat{L}^+(z), L(z) = cL^+(z), \tilde{L}(z) = c\tilde{L}^+(z)$ (for some real positive scalar $c$) for almost all $\bm z \in \mathbb{T}^m$. Similarly, given almost any $\bm z \in \mathbb{T}^m$, solving (\ref{eq:sampled_poly_eqs}) (or equivalently (\ref{eq:matrix_eqs_phaseless})) gives that $\hat{L}(z) = c\hat{L}^+(z), L(z) = cL^+(z), \tilde{L}(z) = c\tilde{L}^+(z)$ for almost all $\bm \theta^+ \in \mathbb{T}^s$ and $\bm g^+ \in \mathbb{C}^s$ (with $s$ non-zero elements).
\end{theorem}
\begin{proof}
Herein, we will denote $\bm G(\bm z)$ in (\ref{eq:matrix_eqs_phaseless}) as $\bm G(\bm z, \bm \theta^+, \bm g^+)$ to reflects its dependency on $\bm \theta^+$, $\bm g^+$ as well. A one-dimensional null space of $\bm G(\bm z, \bm \theta^+, \bm g^+)$ implies that there exists a $(8s-3)\times (8s - 3)$ sub-matrix in $\bm G(\bm z, \bm \theta^+, \bm g^+)$ of full rank $8s - 3$; this sub-matrix is obtained by removing a column from $\bm G(\bm z, \bm \theta^+, \bm g^+)$. Let $d_k(\bm z, \bm \theta^+, \bm g^+)$ denote the determinant of the $(8s-3)\times (8s-3)$ sub-matrix obtained after removing column $k$, $1 \leq k \leq 8s-2$, from $\bm G(\bm z, \bm \theta^+, \bm g^+)$. Hence, the statement that $\bm G(\bm z, \bm \theta^+, \bm g^+)$ has a null space of at least two dimensions is equivalent to $d_k(\bm z, \bm \theta^+, \bm g^+) = 0, 1 \leq k \leq 8s-2$. From (\ref{eq:y_laurent_polys}), it follows that $y_j^+$, and thus $d_k(\bm z, \bm \theta^+, \bm g^+)$ as well, is a ratio of two multivariate Laurent polynomials in the variables $\bm z, \bm \theta^+, \bm g^+$.

First, we show that $\bm G(\bm z, \bm \theta^+, \bm g^+)$ has a one-dimensional null space for almost all $\bm z \in \mathbb{T}^{8s-3}$ given a certain realization $\bm \theta^+ = \hat{\bm \theta}^+$ and $\bm g^+ = \hat{\bm g}^+$ (with $s$ non-zero elements). Theorem \ref{thm:dft_msrmt_unique} shows that for almost all $\bm z$ that contain shifted harmonics (with different shifts) as specified in the theorem, $\bm G(\bm z, \hat{\bm \theta}^+, \hat{\bm g}^+)$ has a one dimensional null space; let $\bm z = \hat{\bm z}$ be one such vector. Hence, $d_j(\hat{\bm z}, \hat{\bm \theta}^+, \hat{\bm g}^+) \neq 0$ for some $j$, $1 \leq j \leq 8s - 2$, implying that $d_j(\bm z, \hat{\bm \theta}^+, \hat{\bm g}^+)$ is not zero for all $\bm z \in \mathbb{T}^{8s-3}$. Consider the zero set $\mathcal{S}_j = \{\bm z | \bm z \in \mathbb{T}^{8s-3}, d_j(\bm z, \hat{\bm \theta}^+, \hat{\bm g}^+) = 0\}$. Since $\mathbb{T}^{8s-3}$ is homeomorphic to the $8s - 3$ dimensional cube $\mathbb{U}^{8s-3}$ (with the obvious homeomorphism $e^{i\pi\alpha} \in \mathbb{T}, \alpha \in \mathbb{U} = [-1, 1)$), $\mathcal{S}_j$ is either empty or a set of points in $\mathbb{T}^{8s-3}$ corresponding to points from a hypersurface of dimension $8s-4$ in the ambient space $\mathbb{U}^{8s-3}$. In either case, $\mathcal{S}_j$ has measure zero, implying that for almost all $\bm z \in \mathbb{T}^{8s-3}$,  $d_j(\bm z, \hat{\bm \theta}^+, \hat{\bm g}^+) \neq 0$ and thus $\bm G(\bm z, \hat{\bm \theta}^+, \hat{\bm g}^+)$ has a one-dimensional null space. Hence, for any given $\bm \theta^+$ and $\bm g^+$ (with $s$ non-zero elements), $\bm G(\bm z, \bm \theta^+, \bm g^+)$ has a one-dimensional null space for almost all $\bm z \in \mathbb{T}^{8s-3}$.

We now show that for almost any given $\bm z \in \mathbb{T}^{8s-3}$, $\bm G(\bm z, \bm \theta^+, \bm g^+)$ has a one-dimensional null space for almost all $\bm \theta^+ \in \mathbb{T}^s$ and $\bm g^+ \in \mathbb{C}^s$ (with $s$ non-zero elements). Since almost all $\bm z \in \mathbb{T}^{8s-3}$ do not belong to the set $\mathcal{S}_j$ above, choose any such element $\tilde{\bm z} \not \in \mathcal{S}_j$. Hence, $d_j(\tilde{\bm z}, \hat{\bm \theta}^+, \hat{\bm g}^+) \neq 0$, implying that $d_j(\tilde{\bm z}, \bm \theta^+, \bm g^+)$ is not zero for all $\bm \theta^+ \in \mathbb{T}^s$ and $\bm g^+ \in \mathbb{C}^s$. Consider the set $\mathcal{T}_j = \{\bm \theta^+, \bm g^+ | \bm \theta^+ \in \mathbb{T}^{s}, \bm g^+ \in \mathbb{C}^{s}, d_j(\tilde{\bm z}, \bm \theta^+, \bm g^+) = 0\}$. This is a subset of the set $\mathbb{T}^{8s-3}\times \mathbb{C}^s$, where $\times$ here denotes the Cartesian product between sets. From above, it follows that $\mathbb{T}^{8s-3}\times \mathbb{C}^s$ is homeomorphic to the set $\mathbb{U}^{8s-3}\times \mathbb{R}^{2s}$ of (real) dimension $8s - 3 + 2s = 10s - 3$ (the complex numbers $\mathbb{C}$ are represented by $\mathbb{R}^2 = \mathbb{R} \times \mathbb{R}$). Hence, $\mathcal{T}_j$ is either empty or a set of points in $\mathbb{T}^{8s-3}\times \mathbb{C}^s$ corresponding to points from a hypersurface of dimension $8s - 3 + 2(s-1) = 10s - 5$ in the ambient space $\mathbb{U}^{8s-3}\times \mathbb{R}^{2s}$. In both cases, it has measure zero in $\mathbb{U}^{8s-3}\times \mathbb{R}^{2s}$, implying that $\bm G(\tilde{\bm z}, \bm \theta^+, \bm g^+)$ has a one-dimensional null space for almost all $\bm \theta^+, \bm g^+$. 
\end{proof}
Theorem \ref{thm:gen_phaseless_samples_theta_g} is true for almost all $\bm z \in \mathbb{T}^m$ with $m \geq 8s - 3$ as well, since almost all subsets of $8s - 3$ samples from $\bm z$ ($8s-3$ rows from $\bm G(\bm z)$) satisfy Theorem \ref{thm:gen_phaseless_samples_theta_g}.

Similar to Corollary \ref{corr:unique_g_abs_4s} we have
\begin{corollary}
\label{corr:unique_g_abs_8s}
Given the realizations $\bm \theta^+, \bm g^+$ of $\bm \theta, \bm g$, respectively. For almost all vectors $\bm z \in \mathbb{T}^m$ with $m \geq 8s - 3$, $\bm \theta^+$ and $|\bm g^+|$ can be found, up to a real positive scalar $c$, with $O(\mathrm{poly}(s))$ complexity.
\end{corollary}
\begin{proof}
Theorem \ref{thm:gen_phaseless_samples_theta_g} shows that for almost all $\bm z$, one obtains $\bm \theta^+$ and $L(z) = cL^+(z)$. 
Using $(\theta_k^+)^{-1} = \ol{\theta_k^+}$ and $$L(\ol{\theta_k^+}) = cL^+(\ol{\theta_k^+}) = 2c|g_k^+|^2\left|t_{k}^+\left(\ol{\theta_k^+}\right)\right|^2$$ it follows that $$c|g_k^+|^2 = \frac{L(\ol{\theta_k^+})}{2\left|t_{k}^+\left(\ol{\theta_k^+}\right)\right|^2}.$$ Since the operations leading to $c|\bm g^+|$ are of $O(\mathrm{poly}(s))$ complexity, it follows that $\sqrt{c}|\bm g^+|$ can be found with $O(\mathrm{poly}(s))$ complexity.
\end{proof}
Next, we show that given $L(z), \tilde{L}(z)$ in (\ref{eq:laurent_polys}), we can solve for $|\tilde{u}(z)|^2, |\hat{u}(z)|^2$.
\begin{lemma}
Given $L(z), \tilde{L}(z)$ in (\ref{eq:laurent_polys}), we have that 
\begin{align*}
|\hat{u}(z)|^2 = \frac{L(z) \pm \sqrt{L^2(z) - 4K(z)}}{2} \\
|\tilde{u}(z)|^2 = \frac{L(z) \mp \sqrt{L^2(z) - 4K(z)}}{2}
\end{align*}
where $K(z) = |\tilde{L}(z)|^2$.
\end{lemma}
\begin{proof}
$L(z) = |\tilde{u}(z)|^2 + |\hat{u}(z)|^2$ implies $|\tilde{u}(z)|^2 = L(z) - |\hat{u}(z)|^2$. Furthermore, $\tilde{L}(z) = \hat{u}(z)\ol{\tilde{u}(z)}$ gives $K(z) = |\tilde{L}(z)|^2 = |\hat{u}(z)|^2|\tilde{u}(z)|^2$. Hence, $K(z) = (L(z) - |\hat{u}(z)|^2)|\hat{u}(z)|^2 = L(z)|\hat{u}(z)|^2 - |\hat{u}(z)|^4$, which is a quadratic equation in $|\hat{u}(z)|^2$. Solving the quadratic, we obtain 
\begin{align}
\label{eq:u_hat_sq_sol}
|\hat{u}(z)|^2 = \frac{L(z) \pm \sqrt{L^2(z) - 4K(z)}}{2}
\end{align}
which together with $L(z) = |\tilde{u}(z)|^2 + |\hat{u}(z)|^2$ gives
\begin{align}
\label{eq:u_tilde_sq_sol}
|\tilde{u}(z)|^2 = \frac{L(z) \mp \sqrt{L^2(z) - 4K(z)}}{2}.
\end{align}
\end{proof}
We have the option to choose sign in (\ref{eq:u_tilde_sq_sol}), resulting in the opposite sign in (\ref{eq:u_hat_sq_sol}). First, assume that $-$ is chosen in (\ref{eq:u_tilde_sq_sol}). Denote the roots of the obtained $|\tilde{u}^+(z)|^2$ (which are distinct according to Lemma \ref{lemma:u_tilde_distinct_roots} and appear as conjugate pairs) as $(r_0,1/\ol{r_0}), \ldots, (r_{s-2}, 1/\ol{r_{s-2}})$. Since $\tilde{L}^+(z) = \hat{u}^+(z)\ol{\tilde{u}^+(z)}$ and due to Theorem \ref{thm:u_hat_tilde_roots_general}, $s-1$ roots of $|\tilde{u}^+(z)|^2$ appear as roots of $\tilde{L}^+(z)$ for almost all $\bm g^+$. Assume WLOG that the roots $r_0,\ldots,r_{s-2}$ appear as roots of $\tilde{L}^+(z)$. Hence, we construct $\tilde{u}^+(z)$ (up to a complex scalar) as $\hat{u}^+(z) = 
\prod_{j=0}^{s-2}(z-r_j)$ (the other $s-1$ roots of $\tilde{L}^+(z)$ produce $\hat{u}^+(z)$). Solving the equations
\begin{align}
\label{eq:u_pol_eqs}
\sum_{k=0}^{s-1}g_kt_{k}^+(r_j) = 0, \quad 0 \leq j \leq s-2
\end{align} 
for $\bm g$ gives a unique $\bm g^m$ up to a complex scalar. Instead, choosing $+$ as sign in (\ref{eq:u_hat_sq_sol}) produces a different $|\tilde{u}^+(z)|^2$ and a different solution $\bm g^p$ (up to a complex scalar) to the equations (\ref{eq:u_pol_eqs}). Since these are the only possible values for $\bm g$, one of them equals $\bm g^+$ (up to a complex scalar). If no prior information about $\bm g^+$ is at hand, then as before, in order to resolve the ambiguity, we take an additional measurement $y_m$ with a random measurement vector. Interestingly, Theorem \ref{thm:dual_sol} shows that the two possible solutions $\bm g^m$, $\bm g^p$ are duals and according to Corollary \ref{corr:dual_sol}  the roots of the Z-transforms of the resulting solutions $\bm x^m, \bm x^p$ are reflections of each other with respect to $\mathbb{T}$.

An important special case is when $(L^+(z))^2 = 4K^+(z)$, which implies that $|\hat{u}^+(z)|^2 = |\tilde{u}^+(z)|^2 = L^+(z)/2$. Theorem \ref{thm:u_hat_tilde_roots_general} shows that for his to occur with a non-zero probability, $(\theta_k^+)^n$, $0 \leq k \leq s-1$, must all be equal. Since $\bm \theta^+ \in \mathbb{T}^s$, this means that $\theta_k^+, 0 \leq k \leq s-1$, are shifted harmonics and thus $\bm x^+$ is $s$-sparse in a DFT basis. In this case, the situation is the same as in \textbf{R5-Alg-Harmonic} since there are $2^{s-1}$ different ways to choose the roots for $\tilde{u}^+(z)$ (each one of them producing a solution to the equations after normalization); thus, the recovery complexity suddenly becomes exponential in $s$. Hence, the implication of this observation is that when $\bm x^+$ is $s$-sparse in a DFT basis, the recovery complexity of $\bm x^+$ is $O(\mathrm{poly}(s)2^s)$ (as mentioned in 2.b in Section \ref{sec:intro}). Nevertheless, due to Theorem \ref{thm:u_hat_tilde_roots_general}, as soon as not all $(\theta_j^+)^n$, $0 \leq j \leq s-1$, are equal (i.e., at least one frequency is not a shifted harmonic), $(L^+(z))^2 \neq 4K^+(z)$ with probability 1 since then the roots of $\tilde{u}^+(z)$ and $\hat{u}^+(z)$ are different for almost all $\bm g^+$.

If $S$ of the $g_0^+,\ldots,g_{s-1}^+$ are non-zero, it is readily seen that $\bm w^+$ in (\ref{eq:w+}) belongs to the null space of $\bm G(\bm z)$ for $s \geq S$; thus the same iterative procedure as in step 2 of \textbf{R5-Alg-Harmonic} can be used to find $S$.
Hence, we have the following algorithm that recovers $\bm g^+, \bm \theta^+$
\par\noindent\rule{\textwidth}{0.4pt}
\textbf{R5-Alg-General:}
\begin{enumerate}[label=\arabic*.]
\item[] Input: $s, n \geq 4s-1, \bm y^+$, a random vector $\bm a \in \mathbb{C}^n$ and arbitrary $\bm z \in \mathbb{T}^m$ where $m \geq 8s-3$.
\item[] Output: $\bm \theta^+, e^{i\alpha}\bm g^+$ for some unknown $\alpha$
\item Construct the $m\times (8s-2)$ matrix $\bm G(\bm z)$ in (\ref{eq:matrix_eqs_phaseless}).
\item Keep reducing the value of $s$ (and thus the dimension of $\bm G(\bm z)$) by one until $\bm G(\bm z)$ has a one-dimensional null space (assume $s = S$ when this occurs for the first time). Let $\bm w$ denote a vector from its one-dimensional null space.
\item Construct the Laurent polynomials $\hat{L}(z) = \sum_{k=0}^{2S}w_kz^{k-S}, \tilde{L}(z) = \sum_{k=0}^{2S-2}w_{k+2S+1}z^{k-S+1}, L(z) = \sum_{k=0}^{2S-2}w_{k+4S}z^{k-S+1}$, respectively. The conjugate of the unique roots of $\hat{L}(z)$ equal $\bm \theta^+$. $|\bm g^+|$ can now be computed (up to a global real positive scalar) from Corollary \ref{corr:unique_g_abs_8s}.
\item {Let $K(z) = |\tilde{L}(z)|^2$ and $$Q(z) = \frac{L(z) + \sqrt{L^2(z) - 4K(z)}}{2}.$$ We have two cases
\begin{enumerate}[label=\alph*.]
\item {$L^2(z) \neq 4K(z)$ (corresponding to $(\theta_j^+)^n$, $0 \leq j \leq S-1$, not all being equal):
\begin{itemize}
\item Denote by $q_0,\ldots,q_{S-2}$ the $S-1$ roots of $Q(z)$ that appear as roots of $\tilde{L}(z)$ and solve for $\bm g$ the $S-1$ polynomial equations $$\sum_{k=0}^{S-1}g_kt_{k}(q_j) = 0, \quad 0 \leq j \leq S-2.$$ Denote by $\bm g^a$ the obtained solution which is unique up to a complex scalar and construct its dual $\bm g^b$ as shown in Theorem \ref{thm:dual_sol}. Let $\mathcal{G} = \{\bm g^a, \bm g^b\}$. Normalize each $\bm g \in \mathcal{G}$ so that $\bm y^+ = |\bm V^T(\bm z)\bm V(\bm \theta^+)\bm g|^2$.
\end{itemize}
}

\item {$L^2(z) = 4K(z)$ (corresponding to $(\theta_j^+)^n$, $0 \leq j \leq S-1$, all being equal):
\begin{itemize}
\item In this case, $L(z)/2 = |\hat{u}(z)|^2 = |\tilde{u}(z)|^2$. The roots of $L(z)/2$ can be paired as $(l_0, 1/\ol{l_0}),\ldots,(l_{S-2},1/\ol{l_{S-2}})$. From each pair, choose one element. This choice can be done in $2^{S-1}$ different ways and let $\bm q^k$, $0 \leq k \leq 2^{S-1}-1$, where $q_j^k = l_j$ or $q_j^k = 1/\ol{l_j}$ for $0 \leq j \leq S-2$, denote the $k$:th choice.
\item For each $\bm q^k$, $0 \leq k \leq 2^{S-1}-1$, solve for $\bm g$ the $S-1$ polynomial equations $$\sum_{n=0}^{S-1}g_nt_{n}(q_j^k) = 0, \quad 0 \leq j \leq S-2.$$ Denote by $\bm g^k$ the obtained solution which is unique up to a complex scalar. Let $\mathcal{G} = \{\bm g^0,\ldots,\bm g^{2^{S-1}-1}\}$. Normalize each $\bm g \in \mathcal{G}$ so that $\bm y^+ = |\bm V^T(\bm z)\bm V(\bm \theta^+)\bm g|^2$.
\end{itemize}
}
\end{enumerate}
}
\item Let $\bm g^+ = \arg\min_{\bm g \in \mathcal{G}}|y_m - |\bm a^T\bm V(\bm \theta^+)\bm g|^2|$. Output $\bm g^+$ and $\bm \theta^+$.
\end{enumerate}
\par\noindent\rule{\textwidth}{0.4pt}
\subsection{Derivation of \textbf{R3}}
From the derivation of \textbf{R4} and \textbf{R5}, we can easily construct an $\bm A$ and a corresponding recovery algorithm that produces \textbf{R3}.
Choose $\bm z \in \mathbb{T}^m$ and $\bm \theta \in \mathbb{T}^n$ such that $\bm z$ and any subset of $s$ samples from $\bm \theta$ satisfy the assumptions of Theorem \ref{thm:dft_msrmt_support} or Theorem \ref{thm:dft_msrmt_unique}; it is clear from the assumptions of these theorems that such $\bm z$ and $\bm \theta$ are easily found. Let $\bm a \in \mathbb{C}^n$ be a random vector. Construct $\bm A$ in \textbf{R3} as $\bm A = [\bm V(\bm z)^T \bm V(\bm \theta); \bm a^T]$. Denote by $p_0,\ldots,p_{s-1}$ the unknown positions of the $s$ non-zero elements in $\bm x$ and let $\bm g = [g_0,\ldots,g_{s-1}]$ denote these elements. Hence, $\bm V(\bm z)^T \bm V(\bm \theta)\bm x = \bm V(\bm z)^T \bm V(\bm \theta_{\bm p})\bm g$, where $\bm \theta_{\bm p} = [\theta_{p_0},\ldots,\theta_{p_{s-1}}]$. Clearly, $\bm\theta_{\bm p}$ can be seen as the unknown $\bm \theta$ in \textbf{R4}, meaning that the first $m$ measurements in \textbf{R3} correspond to \textbf{R4}. From these measurements, one obtains $\bm\theta_{\bm p}$ and the different $\bm g$ satisfying (\ref{eq:phaseless_msrmt_fourier}). The additional measurement with $\bm a$ enables application of \textbf{R5-Alg-Harmonic} or \textbf{R5-Alg-General} which finally produce the results in \textbf{R3}. Thus, a recovery algorithm for \textbf{R3} is
\par\noindent\rule{\textwidth}{0.4pt}
\textbf{R3-Alg:}
\begin{enumerate}[label=\arabic*.]
\item[] \textit{Input}: $s, n \geq 4s-1, \bm y^+$, a random vector $\bm a \in \mathbb{C}^n$ and $\bm A = [\bm V(\bm z)^T \bm V(\bm \theta); \bm a^T]$. Sample vector $\bm z \in \mathbb{T}^m$, where $m \geq 4s-1$, and $\bm \theta \in \mathbb{T}^n$. 
\item[] \textit{Assumptions}: If $4s - 1 \leq m < \min\{n+1,8s-3\}$, then all samples in $\bm z$ are shifted harmonics and $\bm \theta$ is chosen such that $\bm z$ and any subset of $s$ samples from $\bm \theta$ satisfy the assumptions of Theorem \ref{thm:dft_msrmt_support}.\\ If $m \geq 8s - 3$, $\bm z$ and $\bm \theta$ are chosen such that $\bm z$  and any subset of $s$ samples from $\bm \theta$ satisfy the assumptions of Theorem \ref{thm:dft_msrmt_unique}.
\item[] \textit{Output}: $\bm x$ (up to a global phase)
\item If $4s - 1 \leq m \leq \min\{n+1,8s-3\}$, run algorithm \textbf{R5-Alg-Harmonic} until step 3 from which we obtain $\bm \theta^+$. We have that $$p_j = \arg \min_k |\theta_j^+ - \theta_k|, \quad 0\leq j \leq S-1$$ are the positions of the $S$ non-zero elements in $\bm x$. Let $\bm g = [x_{p_0},\ldots,x_{p_{S-1}}]^T$ be the vector of the unknown non-zero elements from $\bm x$. Continue running \textbf{R5-Alg-Harmonic} from step 3 to obtain $\bm g$ as $\bm g = \bm g^+$ in the end.
\item If $m \geq 8s-3$, run algorithm \textbf{R5-Alg-General} until step 3 from which we obtain $\bm \theta^+$. We have that $$p_j = \arg \min_k |\theta_j^+ - \theta_k|, \quad 0\leq j \leq S-1$$ are the positions of the $S$ non-zero elements in $\bm x$. Let $\bm g = [x_{p_0},\ldots,x_{p_{S-1}}]^T$ be the vector of the unknown non-zero elements from $\bm x$. Continue running \textbf{R5-Alg-General} from step 3 to obtain $\bm g$ as $\bm g = \bm g^+$ in the end.
\end{enumerate}
\par\noindent\rule{\textwidth}{0.4pt}

\section{Acknowledgements}
I thank Dr. Zhibin Yu for proofreading the work.

\end{document}